%% file: main.tex
\documentclass[journal,lettersize]{IEEEtran} 
\ifCLASSOPTIONcompsoc
  \usepackage[nocompress]{cite}
\else
  \usepackage{cite}
\fi
\usepackage[OT1]{fontenc} 
\usepackage{array}
\usepackage{marginnote}
\usepackage{soul}
\usepackage{hyperref}

\usepackage{amsmath}
\usepackage{amssymb}
\usepackage{amsthm}
\usepackage{amsfonts}
\usepackage[linesnumbered,ruled,vlined]{algorithm2e}
\SetAlCapNameFnt{\small}
\usepackage{makecell}
\usepackage{algpseudocode}
\usepackage{verbatim}
\usepackage{multirow}
\usepackage{arydshln}
\usepackage{bm}
\usepackage{mathtools}
\usepackage{environ}
\usepackage{etoolbox}
\usepackage{dsfont}
\usepackage{latexsym} 
\usepackage{comment}
\usepackage[table]{xcolor}
\usepackage{graphicx}
\usepackage[shortlabels]{enumitem}
\usepackage{diagbox}
\usepackage{hyphenat}
\usepackage[font={footnotesize,sf},subrefformat=simple,labelformat=simple,justification=centering]{subcaption}
\usepackage{indentfirst}
\usepackage{setspace}
\usepackage{booktabs}
\usepackage[numbers]{natbib}
\usepackage[capitalize]{cleveref}
\usepackage[font={footnotesize,sf},justification=justified,singlelinecheck=false]{caption}
\usepackage{stfloats}
\usepackage{nicefrac}
\usepackage{physics}
\usepackage{marginnote}

\definecolor{blue}{rgb}{0,0,1}
\definecolor{green}{rgb}{0,1,0}
\definecolor{orange}{rgb}{0.9,0.4,0}

\graphicspath{{./figs/}}

\newtheorem{proposition}{Proposition}
\newtheorem{lemma}{Lemma}

\crefname{algocf}{Alg.}{Algs.}
\Crefname{algocf}{Algorithm}{Algorithms}

\newcommand*\widebar[1]{\overline{#1}}

\DeclarePairedDelimiterX{\kldivx}[2]{\big[}{\big]}{%
  #1\;\delimsize\|\;#2%
}

\newcolumntype{L}[1]{>{\raggedright\let\newline\\\arraybackslash\hspace{0pt}}m{#1}}
\newcolumntype{C}[1]{>{\centering\let\newline\\\arraybackslash\hspace{0pt}}m{#1}}
\newcolumntype{R}[1]{>{\raggedleft\let\newline\\\arraybackslash\hspace{0pt}}m{#1}}

\newlength{\myl}
\let\origequation=\equation
\let\origendequation=\endequation
\RenewEnviron{equation}{
  \settowidth{\myl}{$\BODY$}                       
  \origequation
  \ifdimcomp{\the\linewidth}{>}{\the\myl}
  {\ensuremath{\BODY}}                             
  {\resizebox{\linewidth}{!}{\ensuremath{\BODY}}}  
  \origendequation
}

\title{UPMRI: Unsupervised Parallel MRI Reconstruction via Projected Conditional Flow Matching}

\author{Xinzhe Luo, Yingzhen Li, Chen Qin
\IEEEcompsocitemizethanks{
  \IEEEcompsocthanksitem X. Luo and C. Qin are with Imperial College London Department of Electrical and Electronic Engineering and I-X, London, United Kingdom.
  E-mail: \{x.luo, c.qin15\}@imperial.ac.uk.
  \IEEEcompsocthanksitem Y. Li is with Imperial College London Department of Computing, London, United Kingdom.
  E-mail: yingzhen.li@imperial.ac.uk.
  \IEEEcompsocthanksitem This work was supported by the Engineering and Physical Sciences Research Council (EPSRC) UK grants TrustMRI[EP/X039277/1].
  }
  \thanks{Manuscript received ..; revised ..}
}

\begin{document}

\input{00_abstract}

\maketitle
\IEEEdisplaynontitleabstractindextext
\IEEEpeerreviewmaketitle

\input{01_introduction}

\input{04_related_work}

\input{02_background}

\input{03_method}

\input{05_experiment}

\input{06_conclusion}


\input{07_appendix}

\begin{scriptsize}
  \bibliographystyle{IEEEtranN}
  \bibliography{biblio}
\end{scriptsize}

\end{document}

%% file: 00_abstract.tex
\IEEEtitleabstractindextext{
\begin{abstract}
  Reconstructing high-quality images from substantially undersampled k-space data for accelerated MRI presents a challenging ill-posed inverse problem.
  While supervised deep learning has revolutionized this field, it relies heavily on large datasets of fully sampled ground-truth images, which are often impractical or impossible to acquire in clinical settings due to long scan times.
  Despite advances in self-supervised/unsupervised MRI reconstruction, their performance remains inadequate at high acceleration rates.
  To bridge this gap, we introduce UPMRI, an unsupervised reconstruction framework based on Projected Conditional Flow Matching (PCFM) and its unsupervised transformation.
  Unlike standard generative models, PCFM learns the prior distribution of fully sampled parallel MRI data by utilizing only undersampled k-space measurements. 
  To reconstruct the image, we establish a novel theoretical link between the marginal vector field in the measurement space, governed by the continuity equation, and the optimal solution to the PCFM objective. 
  This connection results in a cyclic dual-space sampling algorithm for high-quality reconstruction.
  Extensive evaluations on the fastMRI brain and CMRxRecon cardiac datasets demonstrate that UPMRI significantly outperforms state-of-the-art self-supervised and unsupervised baselines. 
  Notably, it also achieves reconstruction fidelity comparable to or better than leading supervised methods at high acceleration factors, while requiring no fully sampled training data.
\end{abstract}


}

%% file: 01_introduction.tex
\section{Introduction}
Magnetic Resonance Imaging (MRI) is a cornerstone of modern medical diagnostics, providing exceptional soft-tissue contrast without the use of ionizing radiation. 
However, a significant clinical limitation of MRI is its inherently slow data acquisition speed. 
The time required for a scan is directly proportional to the amount of data that must be acquired in the MRI raw data space, known as k-space. 
The development of multi-coil receiver arrays, replacing single coils with smaller, localized ones \citep{journal/mrm/roemer1990, journal/mrm/sodickson1997}, enabled accelerated data acquisition through parallel imaging. 
Nevertheless, reconstructing high-quality MRI from undersampled k-space measurements constitutes a challenging ill-posed inverse problem.
Key algorithms, SENSitivity Encoding (SENSE) \citep{journal/mrm/pruessmann1999} and GeneRalized Autocalibrating Partially Parallel Acquisitions (GRAPPA) \citep{journal/mrm/griswold2002}, laid the groundwork for this. 
By combining parallel imaging with sparsity-promoting terms, later compressed sensing (CS)-based methods achieve higher acceleration than traditional techniques \citep{journal/mrm/lustig2007, journal/spm/lustig2008}.

Building on this concept, data-driven models, particularly physics-informed ``unrolled'' neural networks that emulate classical iterative optimization methods and leverage prior knowledge encoded by convolutional neural networks (CNNs) \citep{conference/cvpr/ulyanov2018}, attain descent performance in terms of both reconstruction quality and speed \citep{journal/tmi/aggarwal2018, journal/mrm/hammernik2018}.
More recently, advancements in accelerated MRI reconstruction have been pursued through modern generative models.
These generative models, rather than focusing on learning a single point estimate, are designed to approximate the entire probability distribution of high-quality MR images \citep{journal/tmi/mardani2018, journal/tmi/tezcan2018, conference/iclr/song2022, journal/media/chung2022}.
This approach aids in solving the inverse problem by enabling sampling from the posterior distribution $p(\bm{x}\mid\bm{y})$, where $\bm{x}$ represents the target fully sampled image, and $\bm{y}$ includes the undersampled multi-coil k-space measurement.
Although these frameworks can attain high reconstruction accuracy, they require large datasets of fully sampled ground-truth images for training, which are expensive and often unavailable. 
For brevity, we refer to them as supervised methods.

Emerging self-supervised and unsupervised techniques aim to reduce the reliance on fully sampled MRI datasets during training \citep{arxiv/wang2025, journal/mrm/yaman2020, journal/tmi/millard2023, journal/bioeng/millard2024, conference/ipmi/luo2025}.
For example, the self-supervised  data underampling (SSDU) family \cite{journal/mrm/yaman2020, journal/tmi/millard2023, journal/bioeng/millard2024} introduced a method in which the observed undersampled k-space measurements are divided into two separate subsets. 
Then a model-based reconstruction network is trained to reconstruct one subset using the other.
More recently, our previous work \cite{conference/ipmi/luo2025} proposed an unsupervised prior learning framework for single-coil MRI reconstruction. 
Nevertheless, these methods fall short in accuracy when dealing with highly undersampled data, e.g., $8\times$ accelerated multi-coil MRI.

To address these challenges, we propose UPMRI, an unsupervised generative model for parallel MRI reconstruction that requires solely undersampled k-space data for training. 
Drawing inspiration from the generalized Stein's Unbiased Risk Estimator (GSURE) \citep{journal/ans/stein1981, journal/tsp/eldar2008}, we present the projected conditional flow matching (PCFM) objective alongside its unsupervised adaptation, which facilitates learning the prior distribution of fully sampled MRI using only undersampled data. 
Since a closed-form projection operator is intractable for parallel MRI, we employ a numerical approximation during training to project the flow matching objective onto the range space of the forward operator.
Furthermore, we derive a new connection between the probability flow in the measurement space under projection and the optimal solution to the proposed PCFM objective.
This facilitates a new inference strategy: a dual-space cyclic integration algorithm that seamlessly couples measurement-space consistency with image-space prior sampling. 
The contributions of this work are summarized as follows:
\begin{itemize}
  \item We propose PCFM, a novel objective that learns the prior distribution of fully sampled multi-coil MRI data using only undersampled measurements, eliminating the need for ground truth.
  \item We derive a theoretical connection between the optimal PCFM solution and probability vector fields in measurement space, enabling a robust dual-space cyclic reconstruction algorithm.
  \item We validate the framework on large-scale public brain (fastMRI) and cardiac (CMRxRecon) datasets. Our results demonstrate that UPMRI sets a new state-of-the-art for unsupervised reconstruction and outperforms several supervised baselines.
\end{itemize}

The remainder of the article is organized as follows.
\cref{sec:related_work} discusses the related work to this study.
\cref{sec:background} states the problem formulation and presents key concepts of flow matching for generative modeling.
\cref{sec:method} describes the proposed projected conditional flow matching objective and the dual-space cyclic flow integration algorithm.
\cref{sec:experiment} presents the experimental setups and evaluation results of our method.
\cref{sec:conclusion} discusses the proposed framework and concludes the study.

%% file: 04_related_work.tex
\section{Related Work}\label{sec:related_work}

\subsection{Diffusion model \& flow matching for inverse problems}
Diffusion models and flow matching constitute two complementary classes of deep generative models for approximating complex, high-dimensional data distributions.
Diffusion models learn the score field $\nabla_{\bm{x}}\log p_t(\bm{x}_t)$ across a continuum of noise levels parametrized by $t$ via denoising score matching, and generate samples by integrating either the reverse-time stochastic differential equation (SDE) or its probability-flow ordinary differential equation (ODE) \cite{conference/nips/ho2020, conference/iclr/song2021}.
Flow matching, by contrast, learns a time-dependent vector field $\bm{v_{\theta}}(\bm{x}_t, t)$ that transports a simple base distribution to the data distribution by solving a continuity equation, enabling deterministic sample generation with ODE solvers \cite{conference/iclr/lipman2023}.

For linear inverse problems with measurement model $\bm{y}=\bm{Ax}+\bm{e}$, these generative models are used to encode a powerful data-driven prior $p(\bm{x})$ learned from large training datasets.
At inference, conditioning measurements $\bm{y}$ enables posterior sampling under $p(\bm{x}\mid\bm{y})\propto p(\bm{y}\mid\bm{x})p(\bm{x})$.
In score-based formulations, the posterior score naturally decomposes as 
\begin{equation}
  \nabla_{\bm{x}_t}\log p_t(\bm{x}_t\mid\bm{y}) = \nabla_{\bm{x}_t}\log p_t(\bm{x}_t) + \nabla_{\bm{x}_t}\log p_t(\bm{y}\mid\bm{x}_t),
\end{equation}
where the prior score is provided by the trained diffusion or flow matching model, and the likelihood score enforces measurement consistency.
Comprehensive surveys of diffusion-based approaches to inverse problems are provided in \cite{arxiv/daras2024, arxiv/chung2025}.

Within this paradigm, diffusion posterior sampling (DPS) approximates the likelihood score with the Tweedie's formula and uses a score-based SDE sampler \cite{conference/iclr/chung2023}; however, practical configurations typically require a large number of function evaluations, i.e., $\operatorname{NFE}\approx 1000$, to achieve high-fidelity reconstruction.
To improve computational efficiency, the denoising diffusion null-space model (DDNM) proposes to enforce measurement consistency by range-null-space decomposition during deterministic DDIM-style sampling and can reduce the sampling cost to roughly $\operatorname{NFE}\approx 100$ while maintaining high reconstruction quality \cite{conference/iclr/wang2023, conference/iclr/ddim2021}.

Flow matching has recently been extended to inverse problems, offering deterministic transport with competitive quality and improved stability under coarse solver discretizations \cite{journal/tmlr/pokle2024, conference/neurips/zhang2024, conference/cvpr/qin2025, conference/iclr/martin2025, conference/iclr/yanfig2025}.
For example, one line of work combines a flow-matching prior velocity with a likelihood score derived from $\Pi$GDM \cite{journal/tmlr/pokle2024, conference/iclr/song2023}.
Another line integrates plug-and-play (PnP) iterations with flow matching to alternate between learned flow-driven denoising and explicit data-consistency gradient steps, thereby improving performance across various inverse problems \cite{conference/iclr/martin2025}.

Despite these advances, both diffusion and flow-matching approaches typically require large, diverge collections of ground-truth images to learn an expressive prior.
In magnetic resonance imaging (MRI), aseembling such datasets with fully-sampled k-space is often impractical or prohibitively expensive due to long acquisition times and patient constraints, limiting the scalability of fully supervised generative priors in real-world deployment \cite{journal/mrm/lustig2007, journal/jmri/marques2019}.

\subsection{Self-supervised \& unsupervised MRI reconstruction}
We explicitly differentiate two training paradigms for MRI reconstruction from solely undersampled data:
(1) Self-supervised approaches that deliberately introduce additional subsampling into the already undersampled k-space, using only a subset of the acquired measurements for supervision;
and (2) unsupervised approaches that utilize all available undersampled k-space measurements without introducing further subsampling.
In particular, the SSDU family implements a holdout strategy by partitioning the acquired undersampled k-space into two disjoint subsets and training a network to predict one subset from the other \cite{journal/mrm/yaman2020, journal/tmi/millard2023, journal/bioeng/millard2024}.
This scheme promotes generalization by preventing trivial identity mappings, but it necessarily discards a portion of the already scarce measurement samples during training, which can limit data utilization efficiency in high-acceleration regimes.
Beyond SSDU, multiple recent works pursue prior learning directly from corrupted measurements using deep generative models.
The ambient diffusion line of work proposes to optimize the standard denoising diffusion objective under self-supervised corruptions applied to the measurements themselves \cite{conference/neurips/daras2023, conference/iclr/aali2024, conference/icml/daras2024, arxiv/daras2025}.

Our work is more related to approaches that adapt the Ensemble Stein's Unbiased Risk Estimator (ENSURE) to train diffusion models or flow matching without ground-truth data.
ENSURE constructs an unbiased surrogate of the mean squared error (MSE), enabling unsupervised training directly from measurements \cite{journal/tmi/aggarwal2022}.
Recent works integrate ENSURE with diffusion and flow-matching objectives to learn prior scores/velocities in an unsupervised fashion \cite{journal/tmlr/kawar2024, conference/ipmi/luo2025}. 
However, the ENSURE framework only provides an unbiased estimate under the single-coil MRI setting where a weighted projection operator is in closed-form, allowing the surrogate to be evaluated efficiently within each optimization step.

In this study, we make the necessary and substantial advance over our previous work \cite{conference/ipmi/luo2025} and additional contributions in the following aspects:
\begin{itemize}
  \item \textbf{Generalization to parallel MRI.} 
  We introduce PCFM to handle the more general multi-coil parallel MRI, while our previous work is only applicable to MRI coil-wisely. 
  We propose a numerical approximation strategy using conjugate gradient during training to project the flow matching objective onto the range space of the forward operator.
  Then, an unsupervised training objective is derived by leveraging the GSURE framework, which learns the prior distribution of fully sampled MRI from only undersampled k-space measurements.

  \item \textbf{Rigorous theoretical analysis.} 
  We provide a comprehensive theoretical framework (Propositions 1–4) in this work. 
  We derive a novel connection between the marginal vector field in the measurement space and the optimal solution to the PCFM objective. 
  This theoretical link justifies our proposed dual-space cyclic integration algorithm, ensuring measurement consistency is mathematically grounded in the flow dynamics.

  \item \textbf{New datasets and benchmarks.}
  We have added the CMRxRecon cardiac MRI dataset, a challenging domain with complex anatomical structures. 
  We benchmark against a wider array of state-of-the-art reconstruction methods, including supervised, self-supervised, and unsupervised baselines.

  \item \textbf{State-of-the-art performance.} 
  We demonstrate that UPMRI outperforms the previous conference method by a significant margin (e.g., +8.2/6.7 dB PSNR on $4\times$/$8\times$ accelerated brain MRI) and achieves reconstruction quality comparable to, and in some cases surpassing, supervised diffusion baselines (e.g., DDNM\textsuperscript{+}) while requiring significantly fewer function evaluations (NFE).
\end{itemize}

%% file: 02_background.tex
\section{Preliminaries}\label{sec:background}

\subsection{Parallel MRI Reconstruction}
The fundamental principle of parallel MRI builds on the fact that each coil in an array receives a distinct spatial sensitivity profile, that is, a unique spatially weighted view of the underlying anatomy.
This spatial encoding ability can be used to compensate for the spatial information lost when k-space is undersampled.
Formally, the forward model of parallel (also known as multi-coil) MRI writes 
\begin{equation}
  \bm{y}_0 = \bm{A}_s\bm{x}_0 + \bm{e},
\end{equation}
where $s$ indexes the randomness in the undersampling mask, $\bm{x}_0\in\mathcal{X}\subset\mathds{C}^D$ is the underlying fully sampled complex-valued image, $\bm{y}_0=[\bm{y}_{0,1}^{\intercal},\dots,\bm{y}_{0,C}^{\intercal}]^{\intercal}\in\mathcal{Y}\subset\mathds{C}^{Cd}$ is the acquired k-space measurements from $C$ receiver coils with $d\leq D$, and $\bm{e}\sim\mathcal{CN}(\bm{0},\sigma_0^2\bm{I}_{Cd})$ denotes measurement noise.
In particular, the \emph{coil-combined} forward operator $\bm{A}_s$ is defined as 
\begin{equation}\label{eq:forward_operator}
  \bm{A}_s \triangleq 
  \begin{bmatrix}
    \bm{M}_s\bm{F}\bm{S}_{1} \\
    \vdots \\
    \bm{M}_s\bm{F}\bm{S}_{C}
  \end{bmatrix} \in \mathds{C}^{Cd\times D},
\end{equation}
which is composed of the undersampling mask $\bm{M}_s\in\{0,1\}^{d\times D}$, the discrete Fourier transform $\bm{F}\in\mathds{C}^{D\times D}$, and diagonal matrices $\bm{S}_{c}\in\mathds{C}^{D\times D}$ representing the sensitivity maps.
For parallel MRI with acceleration factor $\alpha\triangleq\nicefrac{D}{d}> 1$, \cref{eq:forward_operator} can be rank-deficient with a non-zero null space, leading to a challenging ill-posed inverse problem.
We adopt the central assumption that every k-space location is observed at least once within the dataset, which entails the following lemma (proof in \cref{app:lemma1}).
\begin{lemma}[\textbf{Measurement completeness}]\label{lemma1}
  The average projection, $\mathbb{E}_s[\bm{P}_s]$, where $\bm{P}_s=\bm{A}_s^+\bm{A}_s$, is invertible when every k-space location has a non-zero probability of being acquired.
\end{lemma}

\subsection{Conditional Flow Matching}
Conditional flow matching (CFM) \citep{conference/iclr/lipman2023} provides a simulation-free technique to learn a continuous normalizing flow \citep{conference/neurips/chen2018} that transforms a base distribution $p_1$ to a target distribution $p_0$.
In this work, we assume $p_1^X=\mathcal{CN}(\bm{0},2\bm{I}_D)$, and $p_0^X$ produces the underlying fully sampled images, where the superscript $X$ indicates that the flow is in the fully sampled image space $\mathcal{X}$.
This transformation can be specified by an ordinary differential equation (ODE) with a time-dependent smooth vector field $\bm{u}_t^X:[0,1]\times\mathds{C}^D\rightarrow\mathds{C}^D$, i.e., $\dd\bm{x}_t=\bm{u}_t^X(\bm{x}_t)\dd t$.
This induces a probability path $p_t^X$ as the push-forward distribution of $p_0^X$ by the ODE dynamics, satisfying the continuity equation \citep{book/villani2009}:
\begin{equation}\label{eq:continuity}
  \frac{\partial p_t^X}{\partial t} + \nabla\cdot(p_t^X\bm{u}_t^X) = 0.
\end{equation}

CFM proposes to construct such a marginal vector field $\bm{u}_t^X$ by introducing the conditional variable $\bm{z}^X\sim q(\bm{z}^X)$ and conditional vector field $\bm{u}_t^X(\bm{x}_t\mid\bm{z}^X)$.
In this paper, we consider the popular choice of $\bm{z}^X=(\bm{x}_0,\bm{x}_1)$ and the independent coupling $q = p_0^X\times p_1^X$ as generalized by \citep{journal/tmlr/tong2023}.
Then, we assume a conditional probability path $p_t^X(\cdot\mid\bm{z}^X)$ that satisfies the boundary conditions $p_0^X=\mathbb{E}_{q(\bm{z}^X)}\left[p_0^X(\cdot\mid\bm{z}^X)\right]$ and $p_1^X=\mathbb{E}_{q(\bm{z}^X)}\left[p_1^X(\cdot\mid\bm{z}^X)\right]$.
One example is the linear interpolation path $p_t^X(\bm{x}\mid\bm{z}^X)=\delta_{a_t\bm{x}_0+b_t\bm{x}_1}(\bm{x})$ with $a_t=1-t$ and $b_t=t$ \citep{conference/iclr/lipman2023, conference/iclr/liu2023}.
This leads to the conditional vector field $\bm{u}_t^X(\bm{x}\mid\bm{z}^X)=a_t'\bm{x}_0+b_t'\bm{x}_1$ by the continuity equation \emph{w.r.t.} $p_t^X(\bm{x}\mid\bm{z}^X)$, where $a_t'\triangleq\frac{\dd a_t}{\dd t}$ and $b_t'\triangleq\frac{\dd b_t}{\dd t}$.
Then by verifying \cref{eq:continuity}, we can show that the marginal vector field
\begin{equation}
  \bm{u}_t^X(\bm{x}) \triangleq \mathbb{E}_{q(\bm{z}^X)}\left[\bm{u}_t^X(\bm{x}\mid\bm{z}^X)\frac{p_t^X(\bm{x}\mid\bm{z}^X)}{p_t^X(\bm{x})}\right]
\end{equation}
generates the probability flow $p_t^X$.
Meanwhile, learning of the flow is achieved by minimizing the conditional flow matching objective:
\begin{equation}\label{eq:CFM}
  \mathcal{L}_{\text{CFM}}(\bm{\theta}) \triangleq \mathbb{E}_{t, q(\bm{z}^X), p_t^X(\bm{x}\mid\bm{z}^X)}\norm{\bm{h_{\theta}}^X(\bm{x},t)-\bm{u}_t^X(\bm{x}\mid\bm{z}^X)}_2^2,
\end{equation}
where $\bm{h_{\theta}}^X(\cdot,t)$ is a network that predicts the $\mathcal{X}$-space marginal vector field.

%% file: 03_method.tex
\section{Methodology}\label{sec:method}
\begin{figure}[t!]
  \centering
  \includegraphics[width=0.6\linewidth]{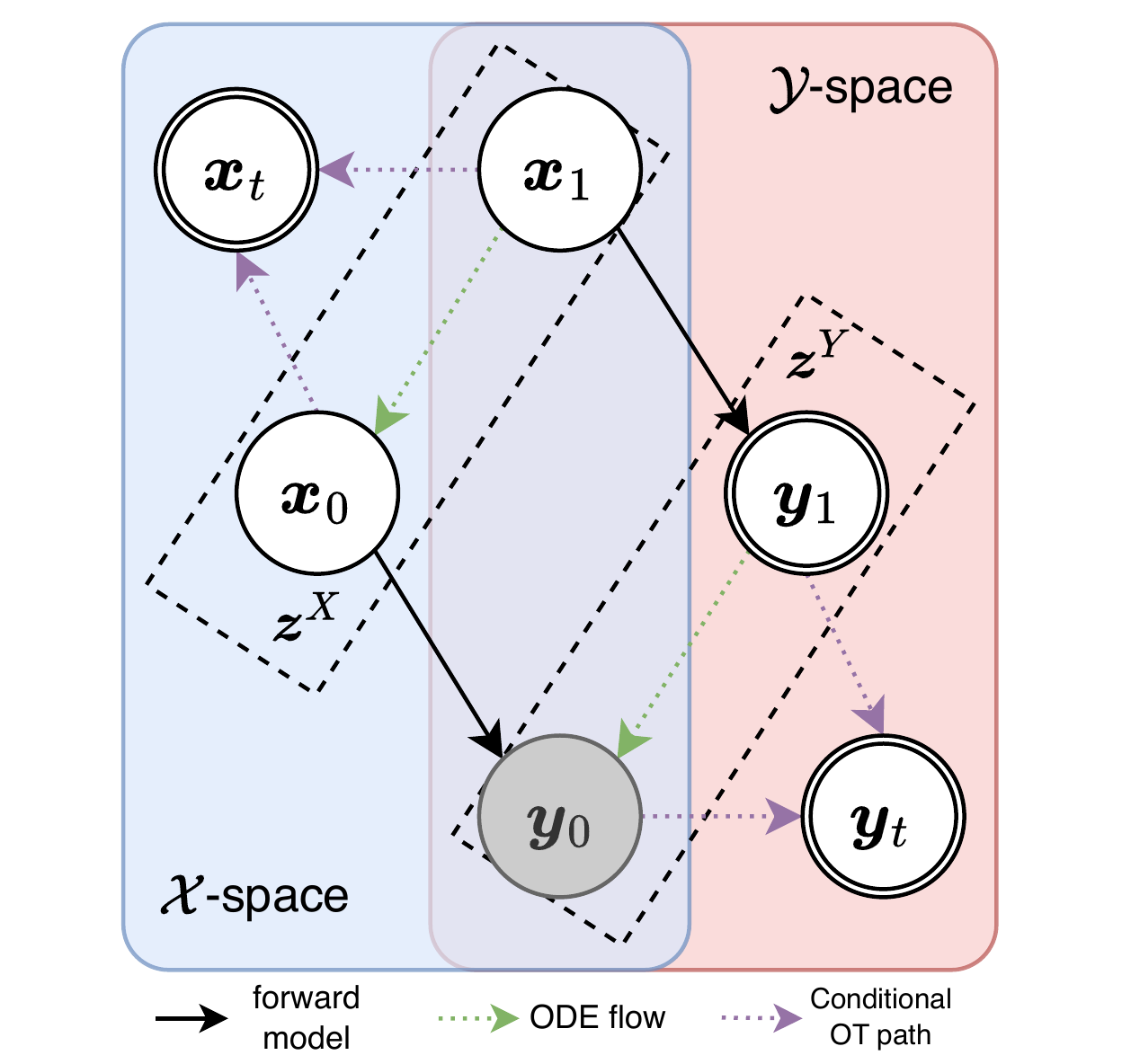}
  \caption{Generation chart of the dual-space conditional probability paths, where observed variables are shaded, and deterministic variables are in double circles.
  Dotted green arrows indicate deterministic ODE flows, whereas purple ones denote conditional probability paths.
  The path from $\bm{x}_1$ to $\bm{y}_0$ is traversed either through the $\mathcal{X}$-space OR the $\mathcal{Y}$-space diagram, as indicated by the background in different colors.
  }
  \label{fig:generative}
\end{figure}

This section introduces the proposed framework designed to reconstruct parallel MRI using only undersampled k-space data. 
The framework integrates two principal components: 
(1) the projected conditional flow matching (PCFM) objective along with its unsupervised transformation as a new formulation for learning the prior (\cref{sec:pcfm}), and 
(2) a new reconstruction algorithm for inference that exploits the relationship between measurement-space probability flow and the optimal solution to PCFM (\cref{sec:cyclic}).

\subsection{Projected Conditional Flow Matching}\label{sec:pcfm}
Due to the scarcity of fully sampled ground-truth signal $\bm{x}_0$, optimizing the $\mathcal{X}$-space CFM objective from \cref{eq:CFM} using a large dataset of fully sampled MRI scans is impractical because the conditional path and vector field within the $\mathcal{X}$-space are not accessible.
To address this, we introduce a $\mathcal{Y}$-space conditional probability path $p_{t}^Y(\bm{y}\mid\bm{z}^Y)=\delta_{a_t\bm{y}_{0}+b_1\bm{y}_{1}}(\bm{y})$, where $\bm{z}^Y\triangleq (\bm{y}_0,\bm{y}_1)$, in which $\bm{y}_{0}=\bm{A}\bm{x}_0+\bm{e}_0$ is the undersampled k-space measurements from parallel MRI, and we define $\bm{y}_{1}\triangleq\bm{A}\bm{x}_1$.
This leads to
\begin{equation}\label{eq:Y_path}
  \begin{aligned}
    p_{t}^Y(\bm{y}\mid\bm{z}^X) &= \int p_{t}^Y(\bm{y}\mid\bm{z}^Y) p(\bm{z}^Y\mid\bm{z}^X)\dd\bm{z}^Y \\
    &= \mathcal{CN}\!\left(\bm{y} \,\middle\vert\, \bm{A}(a_t\bm{x}_0+b_t\bm{x}_1), a_t^2\sigma_0^2\bm{I}_{Cd}\right).
  \end{aligned}
\end{equation}
\cref{fig:generative} depicts the graphical model of the random variables.

Since $\bm{x}_0$ and $\bm{u}_t^X(\bm{x}\mid\bm{z}^X)=a_t'\bm{x}_0+b_t'\bm{x}_1$ are unknown, and the forward operator $\bm{A}$ is rank-deficient, we can only expect to optimize the CFM objective \cref{eq:CFM} in the range space $\mathcal{R}(\bm{A}^{\top})$ of $\bm{A}^{\top}$ \citep{journal/tsp/eldar2008}. 
Therefore, we propose to optimize the following objective function
\begin{equation}\label{eq:PCFM}
  \begin{aligned}
    &\mathcal{L}_{\text{PCFM}}(\bm{\theta}) \triangleq\\
    &\mathbb{E}_{s,t,q(\bm{z}^X),p_t^X(\bm{x}\mid\bm{z}^X),p_{t}^Y(\bm{y}\mid\bm{z}^X)}\norm{\bm{P}_s\!\left[\bm{v_{\theta}}^X(\bm{y},t) - \bm{u}_t^X(\bm{x}\mid\bm{z}^X)\right]}_2^2,
  \end{aligned}
\end{equation}
where $\bm{v_{\theta}}^X(\cdot,t)$ is a network that take as input the $\mathcal{Y}$-space conditional path and approximates the $\mathcal{X}$-space marginal vector field; $\bm{P}_s\triangleq\bm{A}_s^+\bm{A}_s$ is the orthogonal projection onto the range space $\mathcal{R}(\bm{A}_s^{\top})$, with $\bm{A}_s^+$ denoting the Moore-Penrose pseudoinverse of $\bm{A}_s$.
Intuitively, this objective projects the CFM error onto the subspace ``visible'' to the k-space measurements from the mask $\bm{M}_s$.
Thus, we dub this objective function \emph{projected conditional flow matching} (PCFM).
\cref{fig:prior} illustrates the training pipeline of the framework, where we aim to estimate the diffeomorphism from the base distribution to the target distribution of fully sampled MRI using only undersampled measurements.

\begin{figure}[t!]
  \centering
  \includegraphics[width=\linewidth]{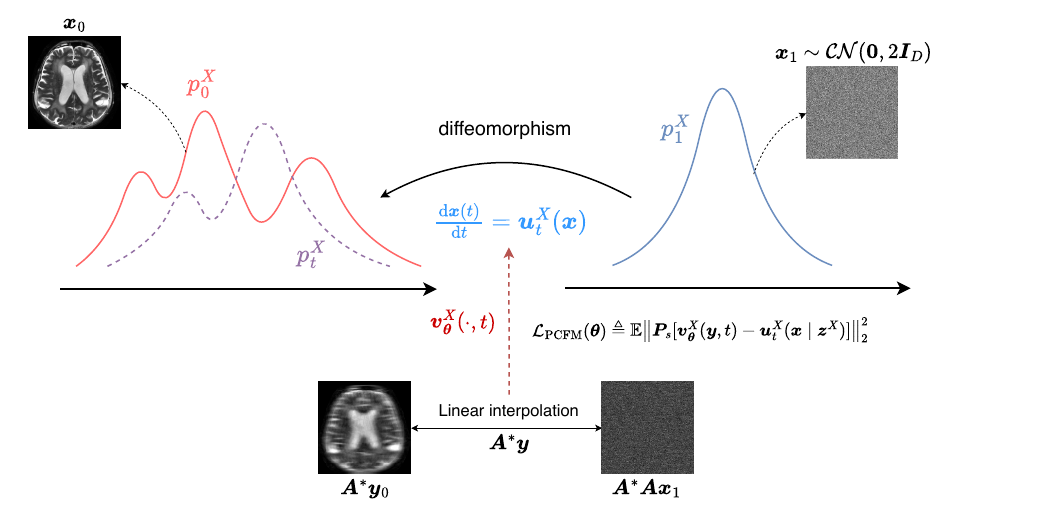}
  \caption{Schematic illustration of the proposed projected conditional flow matching framework.
  The top row depicts the theoretical continuous normalizing flow (diffeomorphism) that transforms samples from a Gaussian base distribution $p_1^X$ (noise $\bm{x}_1$) into the target data distribution $p_0^X$ (fully sampled MRI $\bm{x}_0$) via the probability flow ODE governed by the vector field $\bm{u}_t^X(\bm{x})$.
  The bottom row illustrates the proposed unsupervised training strategy: in the absence of ground-truth fully sampled data, the neural network $\bm{v_{\theta}}^X(\cdot,t)$ approximates the vector field by utilizing a linear interpolation between the zero-filled reconstruction $\bm{A}^*\bm{y}_0$ and the projected noise $\bm{A}^*\bm{Ax}_1$.
  }
  \label{fig:prior}
\end{figure}

We note that within the framework of the parallel MRI forward model, a closed-form solution for this projection is not available. 
To tackle this issue, we note the relationship $\bm{P}=\bm{A}^+\bm{A}=(\bm{A}^*\bm{A})^+\bm{A}^*\bm{A}$ and that $\bm{A}^*\bm{A}$ is a positive semi-definite (PSD) matrix. 
Consequently, we utilize a numerical approximation through the conjugate gradient (CG) method (\cref{app:cg}), which is suitable for solving the linear system $\bm{A}^*\bm{A}\bm{r}_{\bm{\theta},t}^X=\bm{A}^*\bm{A}\left[\bm{v_{\theta}}^X(\bm{y},t)-\bm{u}_t^X(\bm{x}\mid\bm{z}^X)\right]$, where $\bm{r}_{\bm{\theta},t}^X=\bm{P}\left[\bm{v_{\theta}}^X(\bm{y},t)-\bm{u}_t^X(\bm{x}\mid\bm{z}^X)\right]$.
The optimal solution to the PCFM objective is formalized by the following proposition (proof in \cref{app:proposition1}).
\begin{proposition}[\textbf{Optimal solution to PCFM}]\label{proposition1}
  Based on \cref{lemma1}, the minimizer of the PCFM objective is given by
  \begin{equation}
    \bm{v}_{\bm{\theta}^*}^X(\bm{y},t) = \mathbb{E}_{q_{t}(\bm{z}^X\mid\bm{y}),p_t^X(\bm{x}\mid\bm{z}^X)}\left[\bm{u}_t^X(\bm{x}\mid\bm{z}^X)\right],
  \end{equation}
  where $q_{t}(\bm{z}^X\mid\bm{y})=\frac{q(\bm{z}^X)p_{t}^Y(\bm{y}\mid\bm{z}^X)}{p_{t}^Y(\bm{y})}$.
  In particular, when $\bm{u}_t^X(\bm{x}\mid\bm{z}^X)=a_t'\bm{x}_0+b_t'\bm{x}_1$ is independent of $\bm{x}$, we have 
  \begin{equation}\label{eq:opt_solution}
    \bm{v}_{\bm{\theta}^*}^X(\bm{y},t) = \mathbb{E}_{q_{t}(\bm{z}^X\mid\bm{y})}\left[\bm{u}_t^X(\bm{x}\mid\bm{z}^X)\right] = \mathbb{E}\left[\bm{u}_t^X(\bm{x})\mid\bm{y}\right].
  \end{equation}
\end{proposition}
This proposition states that we can learn the expected $\mathcal{X}$-space vector field from only $\mathcal{Y}$-space measurements, which is also the minimum mean squared error estimator of the true $\mathcal{X}$-space marginal vector field.
However, the PCFM objective \cref{eq:PCFM} still depends on the unknown fully sampled MRI $\bm{x}_0$ through Monte Carlo sampling from $q(\bm{z}^X)$ during training.
To address this, inspired by the generalized Stein's unbiased estimator \citep{journal/tsp/eldar2008}, we propose to construct an unbiased estimate of the PCFM objective that does not involve the unknown $\bm{x}_0$.
To this end, we notice an induced linear forward model between the dual-space conditional vector fields
\begin{equation}\label{eq:forward_vec}
  \bm{u}_{t}^Y(\bm{y}\mid\bm{z}^Y) = a_t'\bm{y}_{0} + b_t'\bm{y}_{1} =\bm{A}\bm{u}_{t}^X(\bm{x}\mid\bm{z}^X) + a_t'\bm{e}_{0},
\end{equation}
and the deterministic mapping between the conditional path and the conditional vector field
\begin{equation}
  \bm{y}_t = \frac{a_t}{a_t'}\bm{u}_{t}^Y(\bm{y}_t\mid\bm{z}^Y) - b_t'\left(\frac{a_t}{a_t'}-\frac{b_t}{b_t'}\right)\bm{y}_{1}.
\end{equation}
Based on this, we derive the following unsupervised transformation of the PCFM objective, which does not require fully sampled MRI data for training the vector field predictor (proof in \cref{app:proposition2}).

\begin{proposition}[\textbf{Unsupervised transformation of PCFM}]\label{proposition2}
  Assuming deterministic conditional probability paths $p_t^X(\bm{x}\mid\bm{z}^X)=\delta_{a_t\bm{x}_0+b_t\bm{x}_1}(\bm{x})$ and $p_t^Y(\bm{y}\mid\bm{z}^Y)=\delta_{a_t\bm{y}_0+b_t\bm{y}_1}(\bm{y})$ with $\bm{y}_{0}=\bm{A}_s\bm{x}_{0}+\bm{e}_{0}$ and $\bm{y}_{1}=\bm{A}_s\bm{x}_{1}$, where $\bm{e}_{0}\sim\mathcal{CN}(\bm{0},\sigma_0^2\bm{I}_{Cd})$, then up to a constant the PCFM objective can be transformed to
  \begin{equation}\label{eq:unsupervised}
    \begin{aligned}
      \mathbb{E}_{s,t,q(\bm{z}^Y),p_{t}^Y(\bm{y}\mid\bm{z}^Y)}
    \Big[
      &\norm{\bm{P}_s\!\left[\bm{v_{\theta}}^X(\bm{A}_s^*\bm{y},t)-\widehat{\bm{u}}_{s,t,\text{ML}}^X\right]}_2^2
      \\
      &+ {2a_t}{a_t'}\sigma_0^2\nabla_{\bm{A}_s^*\bm{y}}\cdot\bm{P}_s\bm{v_{\theta}}^X(\bm{A}_s^*\bm{y},t)
    \Big],
    \end{aligned}
  \end{equation}
  where $q(\bm{z}^Y)=q(\bm{y}_0)q(\bm{y}_1)=q(\bm{y}_0)\mathbb{E}_{q(\bm{x}_1)} \left[p(\bm{y}_1\mid\bm{x}_1)\right]$ is sampled by the MRI forward model and Monte Carlo estimation, $\bm{P}_s=\bm{A}_s^+\bm{A}_s$ is the range-space projection, and 
  \begin{equation}
    \widehat{\bm{u}}_{s,t,\text{ML}}^X\triangleq (\bm{A}_s^*\bm{C}_t^{-1}\bm{A}_s)^+\bm{A}_s^*\bm{C}_t^{-1}\bm{u}_{t}^Y(\bm{y}\mid\bm{z}^Y),
  \end{equation}
  where $\bm{C}_t=(a_t'\sigma_0)^2\bm{I}_d$ is the maximum likelihood solution of the forward model in \cref{eq:forward_vec}.
  Note that range-space projection can be approximated by the conjugate gradient method.
\end{proposition}
The network $\bm{v_{\theta}}^X(\cdot,t)$ takes $\bm{A}^*\bm{y}_t$ as input to match the desired dimensionality of the architecture.
The divergence term in \cref{eq:unsupervised} can be estimated using the Hutchinson trace estimator \cite{journal/CSSC/hutchinson1989}, namely,
\begin{equation}
  \nabla_{\bm{A}_s^*\bm{y}_t}\cdot\bm{P}_s\bm{v_{\theta}}^X(\bm{A}_s^*\bm{y},t) = \mathbb{E}_{\bm{b}\sim\mathcal{CN}(\bm{0},\bm{P}_s)}\left[\bm{b}^*\nabla_{\bm{A}_s^*\bm{y}}\bm{v_{\theta}}^X(\bm{A}_s^*\bm{y},t)\bm{b}\right],
\end{equation}
where the expectation is approximated by Monte-Carlo sampling.
By optimizing \cref{eq:unsupervised}, the optimal solution to PCFM can be determined in an unsupervised learning fashion, which learns a fully sampled image prior only from undersampled measurements. 
In the following section, we will delve into reconstructing a fully sampled image utilizing the obtained optimal PCFM solution.

\subsection{Reconstruction via Dual-Space Cyclic Flow Integration}\label{sec:cyclic}

In the inference phase, to reconstruct the fully sampled image $\bm{x}_0$ given the undersampled measurememt $\bm{y}_0$, we note that given the generative model in \cref{fig:generative}, the posterior distribution $p(\bm{x}_0\mid\bm{y}_0)$ can be written as 
\begin{equation}
  \begin{aligned}
    p(\bm{x}_0\mid\bm{y}_0) &= \int p(\bm{x}_0\mid\bm{x}_1,\bm{y}_1,\bm{y}_0) p(\bm{x}_1\mid\bm{y}_1,\bm{y}_0) p(\bm{y}_1\mid\bm{y}_0) \dd \bm{x}_1\dd \bm{y}_1 \\
    &= \int p(\bm{x}_0\mid\bm{x}_1) p(\bm{x}_1\mid\bm{y}_1) p(\bm{y}_1\mid\bm{y}_0) \dd \bm{x}_1\dd \bm{y}_1,
  \end{aligned}
\end{equation}
which can be evaluated by ancestral Monte-Carlo sampling.
\cref{fig:sample} presents the inference pipeline via dual-space cyclic flow integration.
We detail the sampling procedure for each conditional distribution in the following paragraphs.

\begin{figure}[t!]
  \centering
  \includegraphics[width=\linewidth]{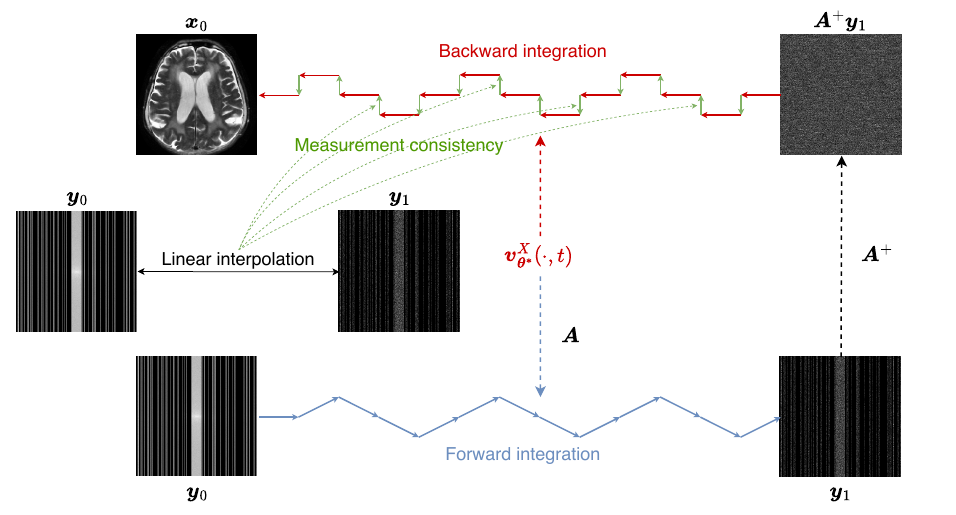}
  \caption{Inference via Dual-Space Cyclic Integration. 
  The algorithm first samples latent noise $\bm{y}_1$ via forward integration in the measurement space (blue). 
  It then reconstructs the target image $\bm{x}_0$ via backward integration (red), enforcing data consistency (green) at each step using the interpolated measurement path.}
  \label{fig:sample}
\end{figure}

\subsubsection{Forward integration and probability flow ODE under projection}
To sample from the distribution $p(\bm{y}_1\mid\bm{y}_0)$, 
we note that given the conditional probability paths in dual spaces illustrated in \cref{fig:generative}, it is feasible to derive a specific relationship between the marginal vector field in the $\mathcal{Y}$-space and the optimal solution to the PCFM as discussed in \cref{proposition1},
leading to a delta distribution of $p(\bm{y}_1\mid\bm{y}_0)$ by the $\mathcal{Y}$-space ODE.
Recall that the $\mathcal{Y}$-space conditional probability path is given by \cref{eq:Y_path}.
Denote $\bm{\mu}_t(\bm{z}^X)\triangleq\bm{A}(a_t\bm{x}_0+b_t\bm{x}_1)$ and $\sigma_t^2 = a_t^2\sigma_0^2$.
Then the time derivative of $p_t^Y(\bm{y}\mid\bm{z}^X)$ can be written as
\begin{equation}
  \begin{aligned}
    &\frac{\partial p_t^Y(\bm{y}\mid\bm{z}^X)}{\partial t} 
    = \frac{\partial p_t^Y(\bm{y}\mid\bm{z}^X)}{\partial \bm{\mu}_t}\cdot \frac{\dd \bm{\mu}_t}{\dd t} + \frac{\partial p_t^Y(\bm{y}\mid\bm{z}^X)}{\partial \sigma_t^2}\frac{\dd\sigma_t^2}{\dd t} \\
    &= -\nabla_{\bm{y}} p_t^Y(\bm{y}\mid\bm{z}^X)\cdot\bm{A}(a_t'\bm{x}_0+b_t'\bm{x}_1) + \frac{1}{2}\Delta_{\bm{y}}p_t^Y(\bm{y}\mid\bm{z}^X) 2a_ta_t'\sigma_0^2 \\ 
    &= -\nabla_{\bm{y}}\cdot\left(p_t^Y(\bm{y}\mid\bm{z}^X)\left[\bm{Au}_t^X(\bm{x}\mid\bm{z}^X)-a_ta_t'\sigma_0^2\nabla_{\bm{y}}\log p_t^Y(\bm{y}\mid\bm{z}^X)\right]\right).
  \end{aligned}
\end{equation}
Therefore, by the continuity equation, we know that the $\mathcal{Y}$-space conditional vector field
\begin{equation}\label{eq:Y_conditional_vec}
  \bm{u}_t^Y(\bm{y}\mid\bm{z}^X) \triangleq \bm{Au}_t^X(\bm{x}\mid\bm{z}^X)-a_ta_t'\sigma_0^2\nabla_{\bm{y}}\log p_t^Y(\bm{y}\mid\bm{z}^X)
\end{equation}
generates the conditional probability path $p_t^Y(\bm{y}\mid\bm{z}^X)$.
Taking the expectation of \cref{eq:Y_conditional_vec} over $q_t(\bm{z}^X\mid\bm{y})$ and leveraging \cref{eq:opt_solution}, we can obtain the marginal vector field in the $\mathcal{Y}$-space in terms of the optimal PCFM solution and the score function, as described in the following lemma (proof in \cref{app:lemma2}).
\begin{lemma}\label{lemma2}
  The $\mathcal{Y}$-space marginal vector field that generates the probability path $p_t^Y$ takes the form
  \begin{equation}
    \bm{u}_t^Y(\bm{y}) = \bm{Av}_{\bm{\theta}^*}^X(\bm{y},t) - a_ta_t'\sigma_0^2\nabla_{\bm{y}}\log p_t^Y(\bm{y}).
  \end{equation}
\end{lemma}
Meanwhile, the relationship between the $\mathcal{Y}$-space marginal vector field $\bm{u}_t^Y(\bm{y})$ and the score function $\nabla_{\bm{y}}\log p_t^Y(\bm{y})$ is established by the following lemma (proof in \cref{app:lemma3}).
\begin{lemma}\label{lemma3}
  Note that $p_1^Y(\bm{y})= \int p_1(\bm{y}\mid\bm{x})p_1^X(\bm{x})\dd\bm{x}= \mathcal{CN}(\bm{y}\mid\bm{0},2\bm{AA}^*)$.
  Then, 
  \begin{equation}
    \bm{u}_t^Y(\bm{y}) = \frac{a_t'}{a_t}\bm{y}-b_t\left(b_t'-\frac{a_t'}{a_t}b_t\right)(2\bm{AA}^*)\nabla_{\bm{y}}\log p_t^Y(\bm{y}).
  \end{equation}
\end{lemma}
The proof is done by writing $\bm{u}_t^Y(\bm{y})$ and $\nabla_{\bm{y}}\log p_t^Y(\bm{y})$ in terms of the conditional expectation $\mathbb{E}_{q_t(\bm{y}_0\mid\bm{y})}[\bm{y}_1]$.
Combining \cref{lemma2} and \cref{lemma3} by canceling out the score function gives the following proposition that relates the $\mathcal{Y}$-space marginal vector field to the optimal solution to PCFM (proof in \cref{app:proposition3}).
\begin{proposition}[\textbf{Vector fields under projection}]\label{proposition3}
  For $a_t=1-t$ and $b_t=t$, the $\mathcal{Y}$-space marginal vector field $\bm{u}_t^Y(\bm{y})$ can be expressed by $\bm{v}_{\bm{\theta}^*}^X(\bm{y},t)$ as
  \begin{equation}
    \begin{aligned}
      \bm{u}_t^Y(\bm{y}) =&\  \bm{Av}_{\bm{\theta}^*}^X(\bm{y},t) \\
      &- \frac{c_t}{1-t}\left[c_t\bm{I}_{Cd}+2\bm{AA}^*\right]^{-1} \left[(1-t)\bm{Av}_{\bm{\theta}^*}^X(\bm{y},t)+\bm{y}\right],
    \end{aligned}
  \end{equation}
  where $c_t\triangleq \frac{(1-t)^2}{t}\sigma_0^2$.
  In addition, left-multiplying both sides with $\bm{A}^*$ gives the more computationally friendly formula when $Cd>D$:
  \begin{equation}
    \begin{aligned}
      \bm{A}^*\bm{u}_t^Y(\bm{y}) =&\  \bm{A}^*\bm{Av}_{\bm{\theta}^*}^X(\bm{y},t) \\
      &- \frac{c_t}{1-t}\left[c_t\bm{I}_{D}+2\bm{A}^*\bm{A}\right]^{-1} \bm{A}^*\left[(1-t)\bm{Av}_{\bm{\theta}^*}^X(\bm{y},t)+\bm{y}\right].
    \end{aligned}
  \end{equation}
\end{proposition}
This proposition asserts that, using the optimal solution of the PCFM, an associated marginal vector field can be derived, which governs the probability flow ODE in the projected measurement space. 
Consequently, to sample from $p(\bm{y}_1\mid\bm{y}_0)$, the $\mathcal{Y}$-space flow can be forward integrated using $\bm{u}_t^Y$ to produce the subsequent $\bm{y}_1$ from the observed initial $\bm{y}_0$, i.e., 
\begin{equation}\label{eq:forward}
  \frac{\dd\bm{y}(t)}{\dd t} = \bm{u}_t^Y(\bm{y}).
\end{equation}

\subsubsection{Posterior sampling at $t=1$}
The posterior distribution $p(\bm{x}_1\mid\bm{y}_1)$ is given by the following lemma (proof in \cref{app:lemma4}).
\begin{lemma}
  Assuming that $p(\bm{x}_1)=\mathcal{CN}(\bm{0}, 2\bm{I}_D)$ and $\bm{y}_1=\bm{Ax}_1$, then we have
  \begin{equation}
    p(\bm{x}_1\mid\bm{y}_1) = \mathcal{CN}\left(\bm{x}_1\mid\bm{A}^+\bm{y}_1, 2(\bm{I}_D-\bm{A}^+\bm{A})\right).
  \end{equation}
\end{lemma}
Therefore, we can draw samples from $p(\bm{x}_1\mid\bm{y}_1)$ via
\begin{equation}\label{eq:posterior}
  \bm{x}_1 = \bm{A}^+\bm{y}_1 + (\bm{I}_D-\bm{A}^+\bm{A})\bm{z}_1,\quad \bm{z}_1\sim\mathcal{CN}(\bm{0},2\bm{I}_D),
\end{equation}
which enables stochastic generation of the reconstruction.

\subsubsection{Backward integration and measurement consistency}
The distribution $p(\bm{x}_0\mid\bm{x}_1)$ is a delta distribution given $\bm{x}_1$ due to the $\mathcal{X}$-space flow, which we can approximate using the backward ODE integration with the PCFM optimal solution, 
\begin{equation}\label{eq:backward}
  \frac{\dd\bm{x}(t)}{\dd t} = \bm{v}_{\bm{\theta}^*}^X(\bm{y}(t),t),
\end{equation}
and for simplicity, we can use $\bm{y}(t)\triangleq a_t\bm{y}_0+b_t\bm{y}_1$.
According to the previous results, we have the following statement (proof in \cref{app:proposition4}):
\begin{proposition}[\textbf{Measurement consistency}]\label{proposition4}
  Given the sample $\bm{x}_1$ from \cref{eq:posterior} and the forward and backward process in \cref{eq:forward} and \cref{eq:backward}, as $\sigma_0\rightarrow 0$, we have
  \begin{equation}
    \bm{Ax}(t) = \bm{y}(t), \quad\forall\, t\in [0,1]
  \end{equation}
  which states that the backward process is consistent with the forward process.
  In particular, we have $\bm{Ax}_0=\bm{y}_0$.
\end{proposition}
Nevertheless, to enforce measurement consistency due to discretization error, we apply the range-null decomposition \citep{conference/iclr/wang2023} after each step of the backward integration,
\begin{equation}
\bm{x}(t) \gets \bm{x}(t) - \bm{A}^+(\bm{A}\bm{x}(t)-\bm{y}(t)),
\end{equation}
where the pseudoinverse operator $\bm{A}^+$ can be approximated by the CG method.

\subsubsection{Reconstruction algorithm}
The overall discrete-time algorithm is outlined in \cref{alg:main}.

\begin{algorithm}[t]
  \caption{Reconstruction via Dual-Space Cyclic Integration with PCFM}
  \label{alg:main}
  \KwIn{k-space measurement $\bm{y}_0$, pretrained optimal solution to PCFM $\bm{v}_{\bm{\theta}^*}^X(\cdot,t)$, number of time steps $T$}
  \KwOut{Reconstructed image $\bm{x}_0$ of $\bm{y}_0=:\bm{y}(0)$}
  \For{$t=0,\dots,\nicefrac{(T-1)}{T}$}{
    $\bm{y} \gets\bm{y}+ \frac{1}{T}\bm{u}_t^Y(\bm{y})$\algorithmiccomment{Forward integration using \cref{proposition3}} \\
  }
  Sample $\bm{x}_1\sim p(\bm{x}_1\mid\bm{y}_1:=\bm{y}(1))$ \algorithmiccomment{Posterior sampling with \cref{eq:posterior}} \\
  \For{$t\in\{\nicefrac{(T-1)}{T}, \dots, 0\}$}{
    ${\bm{y}}\gets a_{t+\nicefrac{1}{T}}\bm{y}_0 + b_{t+\nicefrac{1}{T}}\bm{y}_1$ \\
    $\bm{x}\gets \bm{x}- \frac{1}{T}\bm{v}_{\bm{\theta}^*}^X({\bm{y}},t+\nicefrac{1}{T})$\algorithmiccomment{Backward integration} \\
    $\bm{x} \gets \bm{x} - \bm{A}^+(\bm{A}\bm{x}-{\bm{y}})$\algorithmiccomment{Data consistency step} \\
  }
  \Return{${\bm{x}}_0:=\bm{x}(0)$}
\end{algorithm}

%% file: 05_experiment.tex
\section{Experiments}\label{sec:experiment}

\subsection{Datasets}
Experiments were conducted retrospectively on the NYU fastMRI \citep{journal/rai/knoll2020,arxiv/zbontar2018} and the CMRxRecon challenge (2023) \citep{journal/sd/wang2024,journal/media/lyu2025} datasets to evaluate the performance of the model for accelerated multi-coil MRI reconstruction.
All data used in this work were fully de-identified and anonymized prior to public release. 
Informed consent was obtained from all participating subjects by the original investigators at the time of data acquisition, in accordance with the ethical standards and approvals of the respective data-collecting institutions. 
No identifiable patient information was used or generated in this study.
A selection of 11094/1584/3172 T2-weighted brain MRI slices and 5451/779/1557 cardiac T1/T2 quantitative mapping slices was randomly sampled for training/validation/test, respectively.
Ground-truth images were obtained by the SENSE reconstruction from fully sampled k-space \citep{journal/mrm/pruessmann1999}, i.e., $\bm{x}_0\triangleq\left(\sum_c\bm{S}_c^*\bm{S}_c\right)^{-1}\sum_c\bm{S}_c^*\bm{F}^*\widehat{\bm{y}}_{0,c}$, where $\widehat{\bm{y}}_0$ is the fully sampled k-space measurement, and the coil sensitivity maps were estimated by ESPIRiT \citep{journal/mrm/uecker2014}.
The number of auto-calibration lines (ACLs) from undersampled k-space used for sensitivity map estimation follows the implementation of \cite{journal/mrm/hammernik2021} and \cite{journal/media/lyu2025}.
We retrospectively simulated a random Cartesian (1D) undersampling mask for each image, where every mask contains fully sampled low-frequency k-space lines.
The other lines were randomly uniformly sampled according to the acceleration factor.
The zero-filled adjoint transform $\bm{A}^*\bm{y}_0$ is considered as the undersampled image before reconstruction.

\subsection{Implementation Details}
The linear interpolation coefficients are set as $a_t=t$ and $b_t=1-t$, where $t$ is sampled from a logit normal distribution during training \cite{conference/icml/esser2024}.
The noise level of the original k-space is assumed to be $\sigma_0=10^{-2}$.
We use ADM U-Net \citep{conference/nips/dhariwal2021} as the network architecture for velocity field prediction, where each intermediate convolutional block is followed by adaptive group normalization whose parameters are conditioned on the time points.
Multi-head attention and dropout layers are applied at the lowest two resolutions of the network.
We train the network from scratch on the training data by the AdamW optimizer \citep{conference/iclr/loshchilov2019} for 100K steps, with a learning rate of $10^{-4}$ and a weight decay coefficient of $0.1$.
Exponential moving average of the network parameters
was performed every 100 training steps with a rate of 0.99.
In the inference phase, we set the number of integration steps as $T=10$ and the number of conjugate gradient steps used to approximate the projection operator $\bm{P}$ as $k=10/30$ for training/inference. 

\subsection{Benchmark Study}
To evaluate the efficacy of the proposed UPMRI framework, we conducted comprehensive experiments on two large-scale multi-coil MRI datasetss: the fastMRI brain dataset and the CMRxRecon cardiac T1/T2 mapping dataset.
We compared our unsupervised method against leading reconstruction algorithms across different supervision categories under $4\times$ and $8\times$ acceleration settings.
We utilized two primary evaluation metrics to quantify reconstruction quality: Peak Signal-to-Noise Ratio (PSNR) and Structural Similarity Index (SSIM).

\subsubsection{Experimental setup and baselines}
We compared UPMRI against a wide range of state-of-the-art baselines categorized by their supervision requirements.
\begin{enumerate}[leftmargin=*,label=(\Alph*)]
  \item Supervised methods that require fully sampled MRIs during training. 
  We use the same network architecture and training strategy as our proposed approach for generative model-based methods.
  \begin{itemize}[leftmargin=*]
    \item \textbf{MoDL} \citep{journal/tmi/aggarwal2018} is a model-based end-to-end network that unrolls traditional optimization procedure by regarding the CNN as a regularizer.
    We use 10 unrolling iterations for the network.
    Training is performed by minimizing the L2 loss between the network output and the ground truth.
    \item \textbf{DDNM\textsuperscript{+}} \citep{conference/iclr/wang2023} is a diffusion model-based method for inverse problems solving. 
    Training is performed by denoising on fully sampled images with the DDPM framework \citep{conference/nips/ho2020}, whereas the inference is implemented by alternating between the DDIM sampling steps \citep{conference/iclr/ddim2021} and the range-null decomposition for enforcing measurement consistency using CG iterations \cite{conference/iclr/chung2024}.
    Training is performed by the noise prediction objective with the cosine noise schedule of iDDPM \citep{conference/icml//nichol2021}.
    \item \textbf{OT-ODE} \citep{journal/tmlr/pokle2024} estimates the posterior vector field by combining the original vector field learned from fully sampled images with the likelihood score approximated by the $\Pi$GDM estimation \citep{conference/iclr/song2023}.
    Training is performed by optimizing the original CFM objective \citep{conference/iclr/lipman2023}.
    \item \textbf{PnP-Flow} \citep{conference/iclr/martin2025} integrates the Plug-and-Play framework with flow matching, which alternates between gradient descent steps for measurement consistency, reprojections onto the learned flow path, and denoising by the pre-trained vector field.
    Training is performed by optimizing the original CFM objective \citep{conference/iclr/lipman2023}.
    We set the step size for gradient descent as $\gamma_t=t^{\alpha}$ with $\alpha=0.1$.
  \end{itemize}
  \item Self-supervised methods that learn to reconstruct by additional subsampling of the available k-space measurements.
  We use the VarNet \citep{journal/mrm/hammernik2018} with 10 unrolling iterations as the network backbone for the SSDU family, which we found performed similarly to the MoDL architecture.
  \begin{itemize}[leftmargin=*]
    \item \textbf{SSDU} \citep{journal/mrm/yaman2020} proposes to split the available k-space measurements into two disjoint subsets and then train a model-based reconstruction network to recover one of the subsets from the other.
    Training is achieved by minimizing the L2 loss.
    \item \textbf{Weighted SSDU} \citep{journal/tmi/millard2023} improves upon the SSDU framework by using a subsampling mask of the same distribution as the original mask and re-weighting the L2 loss.
    \item \textbf{Robust SSDU} \citep{journal/bioeng/millard2024} provably recovers clean images from noisy, undersampled training data by simultaneously estimating missing k-space samples and denoising the available samples.
  \end{itemize}
  \item Unsupervised methods that learn to reconstruct using all the available k-space measurements.
  \begin{itemize}[leftmargin=*]
    \item \textbf{REI} \citep{conference/cvpr/chen2022} achieves unsupervised reconstruction by combining the k-space SURE-based loss \citep{journal/ans/stein1981} for measurement consistency and the equivariant imaging framework which builds on the group invariance assumption of the signal space \citep{conference/iccv/chen2021,journal/jmlr/tachella2023}.
    We use the VarNet \citep{journal/mrm/hammernik2018} with 10 unrolling iterations as the network backbone.
    \item \textbf{MOI} \citep{conference/neurips/tachella2022} leverages the randomness in the imaging operator and proposes an unsupervised loss that ensures consistency across all operators.
    \item \textbf{ENSURE} \citep{journal/tmi/aggarwal2022} also leverages the randomness in the forward operators and provides an unbiased estimate of the true mean squared error without fully sampled images. 
    However in the multi-coil setting, it optimizes an approximate version to the true MSE objective.
    \item \textbf{GTF\textsuperscript{2}M} \citep{conference/ipmi/luo2025} proposes a ground-truth-free flow matching framework for single-coil MRI.
    Reconstruction for multi-coil MRI is achieved by coil-wise reconstruction followed by SENSE-based combination.
    We use the same network architecture and training strategy as our proposed approach for this method.
    Nevertheless, this approach can be suboptimal as the prior is learned from coil-wise k-space measurements instead of the combined forward operator of parallel MRI.
  \end{itemize}
\end{enumerate}

\subsubsection{Quantitative results}

\begin{table*}[t]
  \scriptsize
  \centering
  \caption{Quantitative results of 4$\times$ and 8$\times$ accelerated multi-coil MRI reconstruction using various reconstruction methods on the fastMRI brain and CMRxRecon datasets.
  The top and second-best results within each supervision category are highlighted in \textbf{bold} and \underline{underline}, respectively.
  The difference in metrics is statistically significant between our method and the others by the two-sided paired $t$-test ($p<0.05$).}
  \resizebox{1\linewidth}{!}{
  \begin{tabular}{C{1.55cm}C{1.5cm}C{1.8cm}C{1.8cm}c@{\hspace{0.cm}}C{1.6cm}C{1.6cm}c@{\hspace{0.cm}}C{1.8cm}C{1.8cm}c@{\hspace{0.cm}}C{1.6cm}C{1.6cm}C{0.5cm}}
    \toprule
    \multirow{3}{*}{\textbf{\shortstack[c]{Training\\ Supervision}}} & \multirow{3}{*}{\textbf{\shortstack[c]{Reconstruction\\ Method}}} & \multicolumn{5}{c}{fastMRI Brain} & & \multicolumn{5}{c}{CMRxRecon Cardiac T1/T2 Mapping} & \multirow{3}{*}{NFEs $\downarrow$} \\
    \cmidrule{3-7}\cmidrule{9-13}
    & & \multicolumn{2}{c}{SSIM $\uparrow$} & & \multicolumn{2}{c}{PSNR $\uparrow$} & & \multicolumn{2}{c}{SSIM $\uparrow$} & & \multicolumn{2}{c}{PSNR $\uparrow$} & \\
    \cmidrule{3-4}\cmidrule{6-7}\cmidrule{9-10}\cmidrule{12-13}
    & & $4\times$ & $8\times$ & & $4\times$ & $8\times$ & & $4\times$ & $8\times$ & & $4\times$ & $8\times$ & \\
    \midrule
    None & Zero-Filled & $0.815\pm0.087$ & $0.730\pm0.113$ & & $28.28\pm3.80$ & $24.44\pm3.98$ & & $0.769\pm0.057$ & $0.747\pm0.056$ & & $27.01\pm1.97$ & $26.11\pm1.90$ & N/A \\
    \hdashline\noalign{\vskip 0.5ex}
    \multirow{4}{*}{Supervised} & MoDL & $\bm{0.970\pm0.036}$ & $\underline{0.916\pm0.044}$ & & $\underline{39.71\pm2.93}$ & $32.32\pm3.07$ & & $\bm{0.979\pm0.011}$ & $\underline{0.943\pm0.025}$ & & $\underline{41.79\pm3.38}$ & $\underline{36.36\pm3.14}$ & 1 \\
    & DDNM\textsuperscript{+} & $0.938\pm0.041$ & $\bm{0.920\pm0.040}$ & & $\bm{40.00\pm2.78}$ & $\bm{35.24\pm2.85}$ & & $\underline{0.977\pm0.011}$ & $\bm{0.953\pm0.022}$ & & $\bm{44.93\pm3.35}$ & $\bm{39.60\pm3.31}$ & 100 \\
    & OT-ODE & $0.907\pm0.054$ & $0.852\pm0.060$ & & $33.95\pm2.32$ & $28.86\pm2.94$ & & $0.922\pm0.036$ & $0.870\pm0.052$ & & $35.65\pm3.23$ & $32.08\pm3.00$ & 100 \\
    & PnP-Flow & $\underline{0.951\pm0.044}$ & $0.913\pm0.043$ & & $37.92\pm2.44$ & $\underline{32.85\pm2.80}$ & & $0.963\pm0.021$ & $0.924\pm0.038$ & & $39.77\pm3.57$ & $35.13\pm3.46$ & 100 \\
    \hdashline\noalign{\vskip 0.5ex}
    \multirow{3}{*}{Self-supervised} & SSDU & $0.831\pm0.076$ & $0.792\pm0.094$ & & $28.13\pm2.44$ & $26.61\pm3.68$ & & $\underline{0.888\pm0.035}$ & $0.811\pm0.051$ & & $\underline{32.64\pm2.59}$ & $29.23\pm2.37$ & 1 \\
    & W-SSDU & $\underline{0.939\pm0.049}$ & $\bm{0.870\pm0.055}$ & & $\underline{36.09\pm2.64}$ & $\bm{29.72\pm2.90}$ & & ${0.872\pm0.050}$ & $\underline{0.831\pm0.050}$ & & ${32.11\pm2.88}$ & $\underline{29.82\pm2.58}$ & 1 \\
    & R-SSDU & $\bm{0.941\pm0.047}$ & $\underline{0.856\pm0.063}$ & & $\bm{36.23\pm2.58}$ & $\underline{29.27\pm3.15}$ & & $\bm{0.912\pm0.031}$ & $\bm{0.861\pm0.044}$ & & $\bm{33.93\pm2.59}$ & $\bm{31.10\pm2.66}$ & 1 \\
    \hdashline\noalign{\vskip 0.5ex}
    \multirow{5}{*}{Unsupervised} & REI & $0.683\pm0.120$ & $0.715\pm0.118$ & & $21.59\pm2.90$ & $21.62\pm2.76$ & & $0.736\pm0.056$ & $0.716\pm0.060$ & & $23.65\pm2.12$ & $26.62\pm1.96$ & 1 \\
    & MOI & $0.869\pm0.065$ & $0.747\pm0.105$ & & $30.97\pm2.74$ & $25.16\pm3.94$ & & $\underline{0.971\pm0.015}$ & $0.874\pm0.040$ & & $\underline{40.13\pm3.26}$ & $\underline{31.47\pm2.62}$ & 1 \\
    & ENSURE & $0.899\pm0.065$ & $0.800\pm0.104$ & & $32.75\pm4.20$ & $27.12\pm4.33$ & & $0.918\pm0.041$ & $0.849\pm0.052$ & & ${34.57\pm3.13}$ & $30.27\pm2.48$ & 1 \\
    & GTF\textsuperscript{2}M & $\underline{0.916\pm0.055}$ & $\underline{0.852\pm0.057}$ & & $\underline{33.99\pm2.33}$ & $\underline{28.40\pm3.05}$ & & ${0.918\pm0.028}$ & $\underline{0.881\pm0.038}$ & & $33.67\pm2.43$ & ${31.33\pm2.42}$ & 20 \\
    & UPMRI (Ours) & $\bm{0.983\pm0.032}$ & $\bm{0.948\pm0.034}$ & & $\bm{42.19\pm3.71}$ & $\bm{35.08\pm3.35}$ & & $\bm{0.994\pm0.008}$ & $\bm{0.974\pm0.016}$ & & $\bm{52.39\pm9.81}$ & $\bm{40.50\pm3.84}$ & 20 \\
    \bottomrule
  \end{tabular}
  }
  \label{tab:quantitative_results}
\end{table*}

\cref{tab:quantitative_results} presents the quantitative results for both datasets. 
The proposed UPMRI method demonstrates robust superiority, consistently achieving the highest scores among unsupervised methods and frequently outperforming supervised benchmarks.

On the brain dataset, UPMRI achieved a PSNR of 42.19 dB and SSIM of 0.983 at $4\times$ acceleration.
This represents a substantial improvement over the best-performing supervised method, MoDL (39.71 dB, 0.970), and the top self-supervised method, R-SSDU (36.23 dB, 0.941). 
Among unsupervised competitors, UPMRI outperformed the closest rival, GTF\textsuperscript{2}M, by a margin of 8.2 dB.
At the more challenging $8\times$ acceleration, UPMRI maintained exceptional fidelity with a PSNR of 35.08 dB and SSIM of 0.948. 
While the supervised diffusion model DDNM\textsuperscript{+}
achieved a slightly higher PSNR (35.24 dB), UPMRI surpassed it significantly in structural similarity (0.948 vs. 0.920).

The proposed method achieved its more significant gains in cardiac imaging.
At $4\times$ acceleration, UPMRI reached a PSNR of 52.39 dB and SSIM of 0.994.
This reflects a dramatic performance gap; UPMRI outperformed the leading supervised method (DDNM\textsuperscript{+}) by 7.46 dB and the leading unsupervised method (MOI) by 12.26 dB.
At $8\times$ acceleration, UPMRI remained robust, achieving 40.50 dB PSNR and 0.974 SSIM.
It successfully outperformed all supervised candidates, including DDNM\textsuperscript{+} (39.60 dB) and MoDL (36.36 dB), proving that UPMRI does not require fully sampled data to achieve state-of-the-art reconstruction in cardiac scenarios.

Furthermore, our approach demonstrates better efficiency relative to baseline approaches based on supervised diffusion models or flow matching, even though it employs the same network structure and training strategy. 
This is achieved with a significant decrease in NFEs: UPMRI requires only 20 function evaluations, making it several times faster than diffusion- and flow-based methods like DDNM\textsuperscript{+} and OT-ODE (100 NFEs).
\cref{tab:time} illustrates the average inference time of the various methods evaluated for the reconstruction of a single image from the CMRxRecon 2023 dataset. It is evident that among the generative model-based techniques, our method demonstrates the fastest inference time.

\begin{table}
  \scriptsize
  \centering
  \caption{Average inference time of the compared methods for reconstructing one image from the CMRxRecon 2023 dataset.
  The inference time is calculated on an NVIDIA A5000 GPU.}
  \begin{tabular}{C{1.6cm}C{1.9cm}C{1.8cm}C{0.8cm}}
    \toprule
    {\textbf{\shortstack[c]{Training\\ Supervision}}} & {\textbf{\shortstack[c]{Reconstruction\\ Method}}} & \textbf{Time} (ms) & {NFEs $\downarrow$} \\
    \midrule
    \multirow{4}{*}{Supervised} & MoDL & 58 & 1 \\
    & DDNM\textsuperscript{+} & 1000 & 100 \\
    & OT-ODE & 1750 & 100 \\
    & PnP-Flow & 938 & 100 \\
    \midrule
    \multirow{3}{*}{Self-supervised} & SSDU & 19 & 1 \\
    & Weighted SSDU & 19 & 1 \\
    & Robust SSDU & 19 & 1 \\
    \midrule
    \multirow{5}{*}{Unsupervised} & REI & 19 & 1 \\
    & MOI & 58 & 1 \\
    & ENSURE & 58 & 1 \\
    & GTF\textsuperscript{2}M & 1563 & 20 \\
    \rowcolor{gray!20}
    & UPMRI (Ours) & 500 & 20 \\
    \bottomrule
  \end{tabular}
  \label{tab:time}
\end{table}

\subsubsection{Qualitative visualization}

\begin{figure*}[t!]
  \centering
  \includegraphics[width=\textwidth]{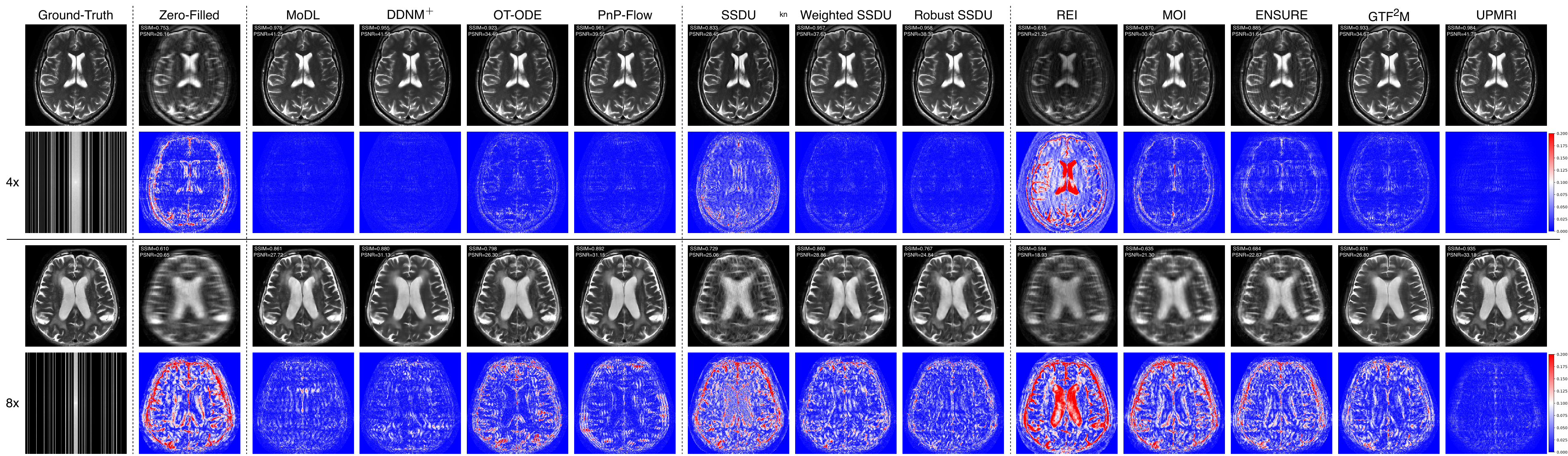}
  \caption{Visualization of reconstruction on two test samples of $4\times$ and $8\times$ accelerated multi-coil brain MRI from the compared methods. The k-space are presented in log-scale absolute values. The error maps are presented in values relative to the peak intensity in the ground-truth image (blue indicates low error, red/white indicates high error).}
  \label{fig:brain_compare}
\end{figure*}

\begin{figure*}[t!]
  \centering
  \includegraphics[width=\textwidth]{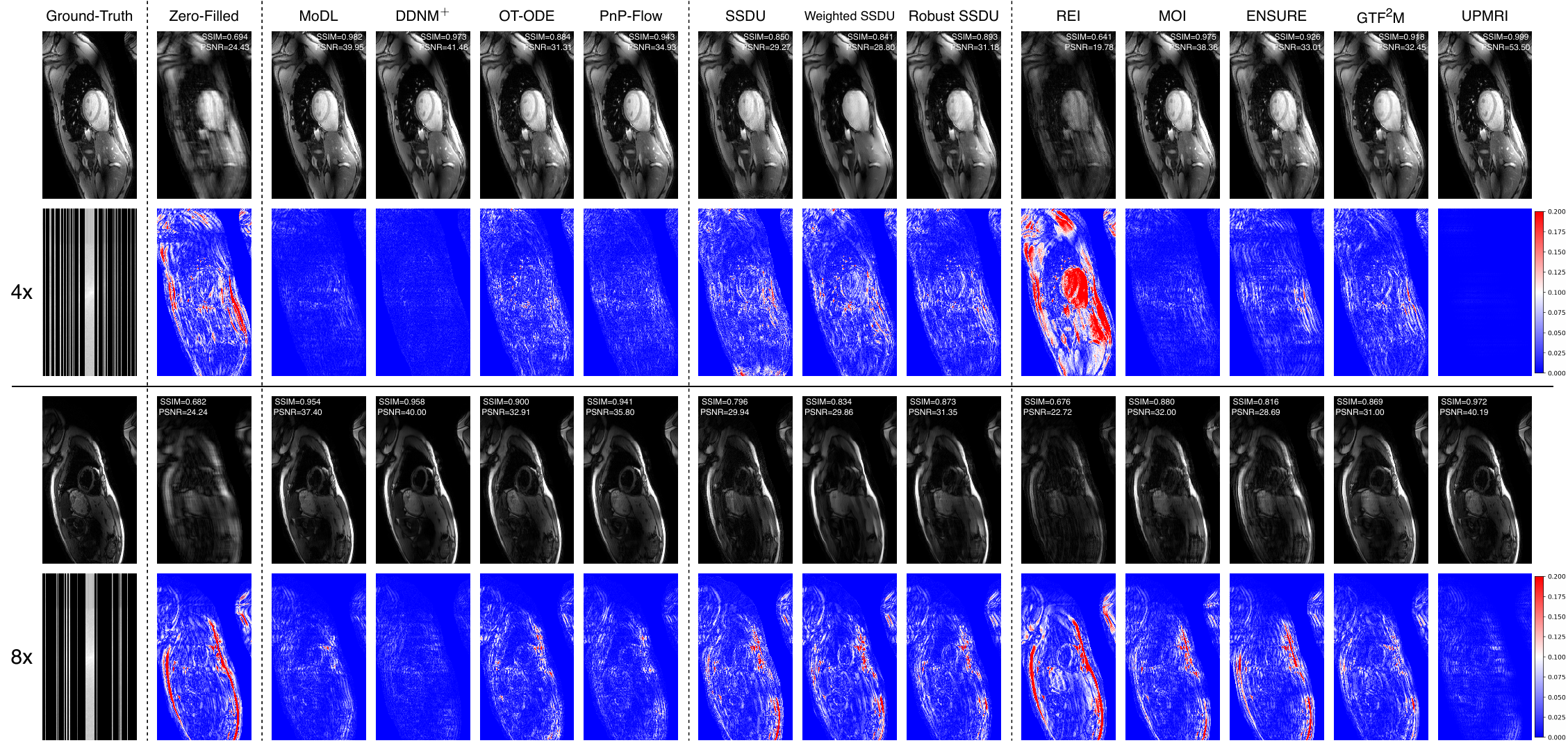}
  \caption{Visualization of reconstruction on two test samples of $4\times$ and $8\times$ accelerated multi-coil cardiac T1/T2-mapping MRI from the compared methods. The k-space measurements are presented in log-scale absolute values. The error maps are presented in values relative to the peak intensity in the ground-truth image (blue indicates low error, red/white indicates high error).}
  \label{fig:cmr_compare}
\end{figure*}

\cref{fig:brain_compare} and \cref{fig:cmr_compare} visualize the reconstruction quality and error distribution for brain and cardiac datasets, respectively. The error maps depict the absolute difference between the reconstruction and the ground truth, where blue indicates low error and red/white indicates high error.

On the brain dataset, at both $4\times$ and $8\times$ acceleration, baseline methods like SSDU and REI display heavy aliasing artifacts, rendering fine details indistinguishable.
The error maps for competitive unsupervised methods such as GTF\textsuperscript{2}M and ENSURE show visible residual structures corresponding to the brain's folding patterns where high-frequency details are missing.
In contrast, UPMRI excels at preserving sharp anatomical boundaries.
In the regions of the brain ventricles and sulci, UPMRI recovers fine cortical details that are lost or blurred in the supervised and self-supervised reconstructions, producing images visually indistinguishable from the ground-truth.

Cardiac quantitative imaging is challenging due to complex signal variations and anatomy.
In \cref{fig:cmr_compare}, unsupervised methods like REI and MOI degrade significantly at $8\times$ acceleration, evidenced by red/white regions in the error maps, representing severe signal deviation.
Consistent with the quantitative results, the UPMRI error maps are the most suppressed among the comparisons, confirming that it achieves the highest pixel-wise accuracy.
Even at $8\times$ acceleration, UPMRI maintains the distinct contrast between the myocardium and the blood pool and removes aliasing.

\subsection{Parameter Study}

\begin{figure}[t!]
  \centering
  \begin{subfigure}{\linewidth}
    \centering
    \caption{Performance vs. number of integration steps}
    \includegraphics[width=\linewidth]{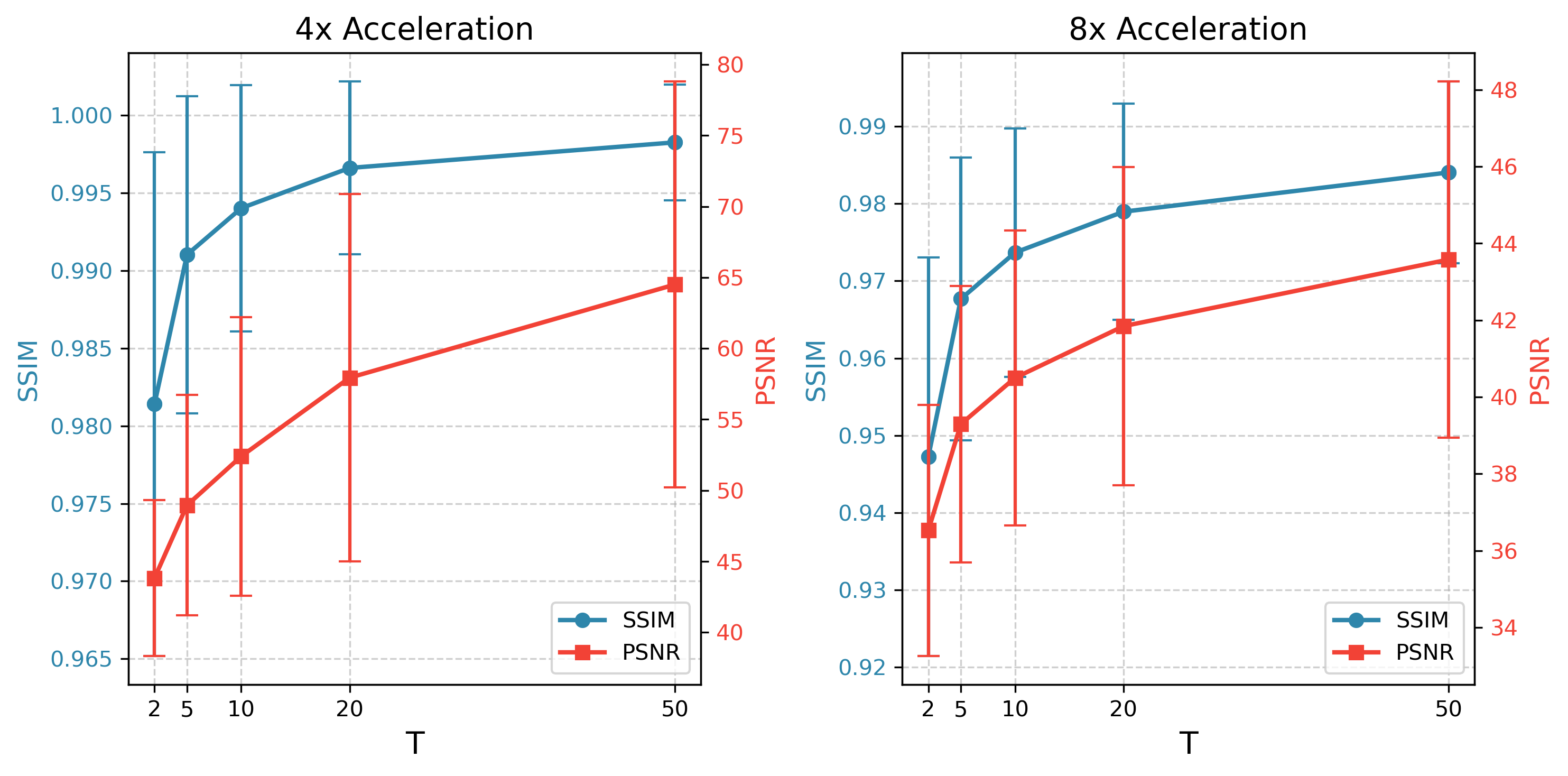}
    \label{fig:int_steps}
  \end{subfigure}
  \vspace{-0.5cm}
  \\
  \begin{subfigure}{\linewidth}
    \centering
    \caption{Performance vs. number of test-time CG iterations}
    \includegraphics[width=\linewidth]{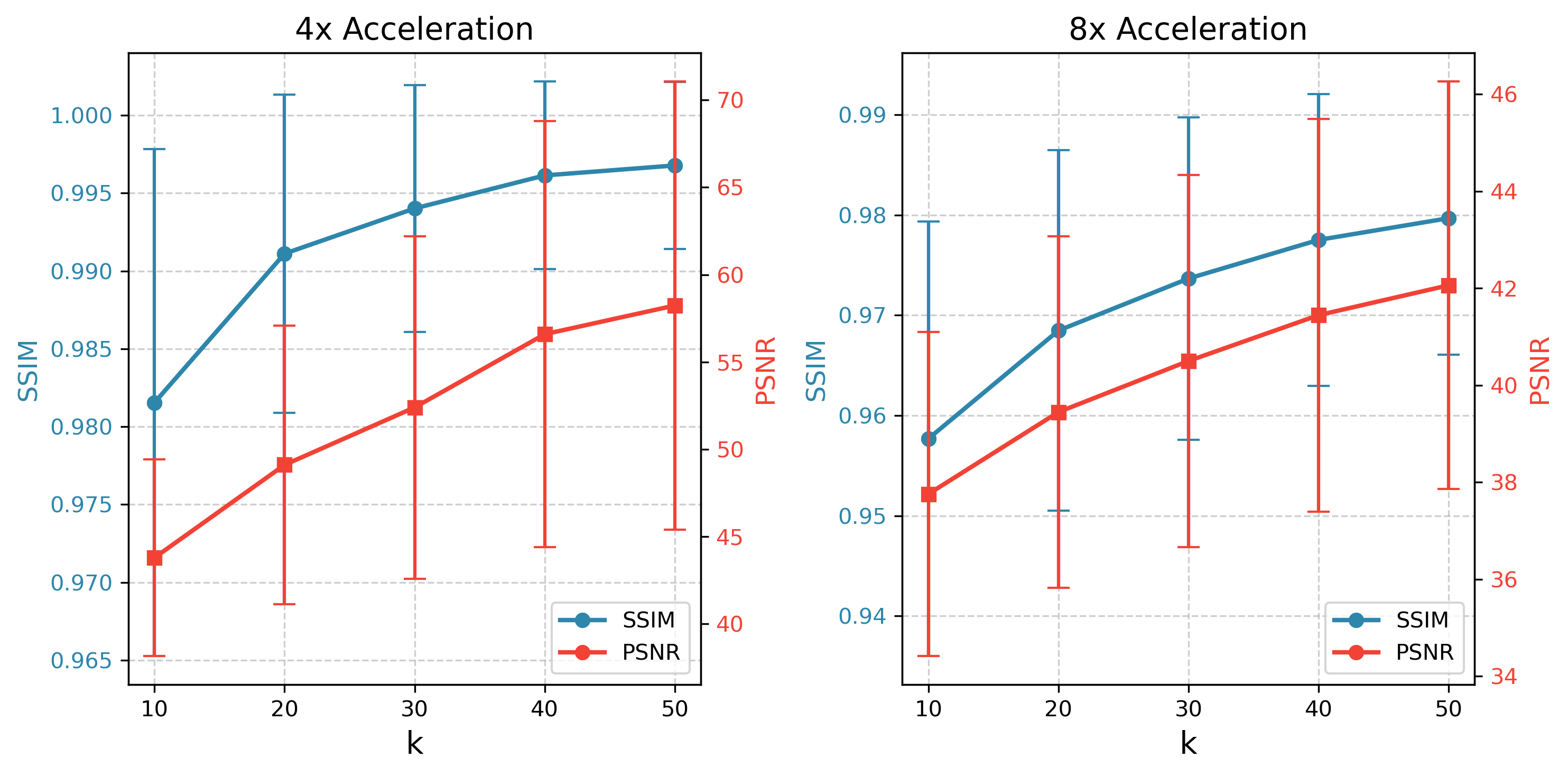}
    \label{fig:cg_steps}
  \end{subfigure}
  \vspace{-0.5cm}
  \caption{Reconstruction performance on the CMRxRecon cardiac mapping dataset ($4\times$ \& $8\times$ acceleration) as a function of the number integration steps and test-time CG iterations, respectively. 
  Note that the default values are $T=10$ and $k=30$.}
  \label{fig:param_study}
\end{figure}

To evaluate the sensitivity of our proposed framework to its hyperparameters, we conducted parameter studies on the CMRxRecon cardiac dataset.
Specifically, we analyzed the impact of the number of integration steps ($T$) and the number of test-time CG steps ($k$) on reconstruction performance, measured by SSIM and PSNR.

\cref{fig:int_steps} illustrates the reconstruction performance as a function of the number of integration steps $T$, ranging from 2 to 50, for both $4\times$ and $8\times$ acceleration. 
We observe a clear logarithmic trend where reconstruction quality improves rapidly as $T$ increases from 2 to 10. 
However, beyond $T=10$, the SSIM curve begins to plateau, indicating diminishing returns in structural fidelity, although PSNR continues to show slight linear gains. 
While $T=50$ yields the highest numerical metrics, it incurs a linearly proportional increase in computational cost. 
Consequently, we selected $T=10$ as the default setting, as it represents an optimal trade-off point that achieves high-fidelity reconstruction while maintaining rapid inference speeds (20 NFEs).

\cref{fig:cg_steps} analyzes the effect of the number of CG iterations $k$ used to approximate the projection operator during the data consistency steps. 
We varied $k$ from 10 to 50. The results show a consistent improvement in both SSIM and PSNR as $k$ increases, confirming that more accurate projection onto the measurement subspace leads to better restoration. 
Similar to the integration steps, the performance gains begin to saturate around $k=30$, particularly for the SSIM metric at $4\times$ acceleration. 
To ensure rigorous enforcement of measurement consistency without unnecessary computational overhead, we set the default value to $k=30$.

These trends remain consistent across both $4\times$ and $8\times$ acceleration settings, demonstrating the robustness of our parameter choices across different undersampling regimes.

\subsection{Generalization to Unseen Acceleration Rates}

\begin{figure}[t!]
  \centering
  \includegraphics[width=\linewidth]{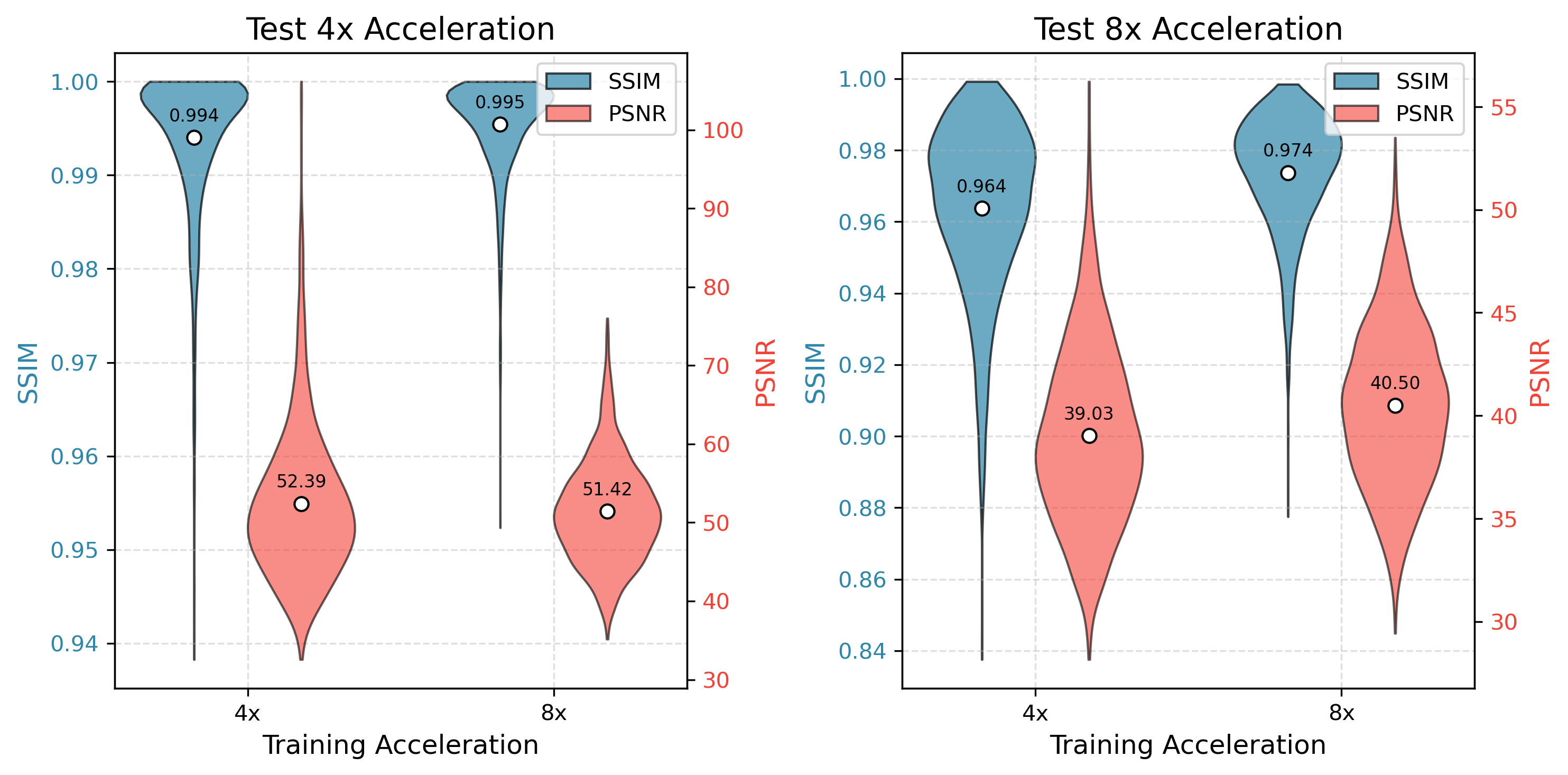}
  \caption{Generalization to unseen acceleration rates. 
  The violin plots display the distribution of SSIM (blue, left axis) and PSNR (red, right axis) scores. 
  (Left) Evaluation on $4\times$ test data shows that the model trained on $8\times$ data maintains performance levels comparable to the model trained on $4\times$ data. 
  (Right) Evaluation on $8\times$ test data demonstrates that the model trained on $4\times$ data generalizes well to higher acceleration rates, matching the performance of the model trained on $8\times$ data. 
  These results confirm the model's robustness and ability to generalize to unseen sampling masks.}
  \label{fig:gen_acc_4_8}
\end{figure}

To evaluate the robustness of the proposed UPMRI framework, we conducted a cross-evaluation experiment to assess how well the model generalizes to acceleration rates not seen during training.
We used the models trained separately on $4\times$ and $8\times$ accelerated k-space data and evaluated them on both test sets.
As illustrated in \cref{fig:gen_acc_4_8}, the model demonstrates good consistency regardless of the training data acceleration.
The left panel shows the evaluation on $4\times$ accelerated test data.
Notably, the model trained on the more challenging $8\times$ data maintains high performance that are comparable to the model trained on the matched $4\times$ data.

Conversely, the right panel displays the results on the highly undersampled $8\times$ test set.
Surprisingly, the model trained on the $4\times$ acceleration generalizes effectively to this harder task, producing SSIM and PSNR distributions that are similar to the model trained directly on $8\times$ data.
The similar shape and spread of the violin plots across different training configurations indicate that UPMRI effectively learns the underlying data distribution rather than overfitting to specific sampling masks, thereby demonstrating strong generalizability to unseen sampling patterns.

\begin{figure}[t!]
  \centering
  \includegraphics[width=\linewidth]{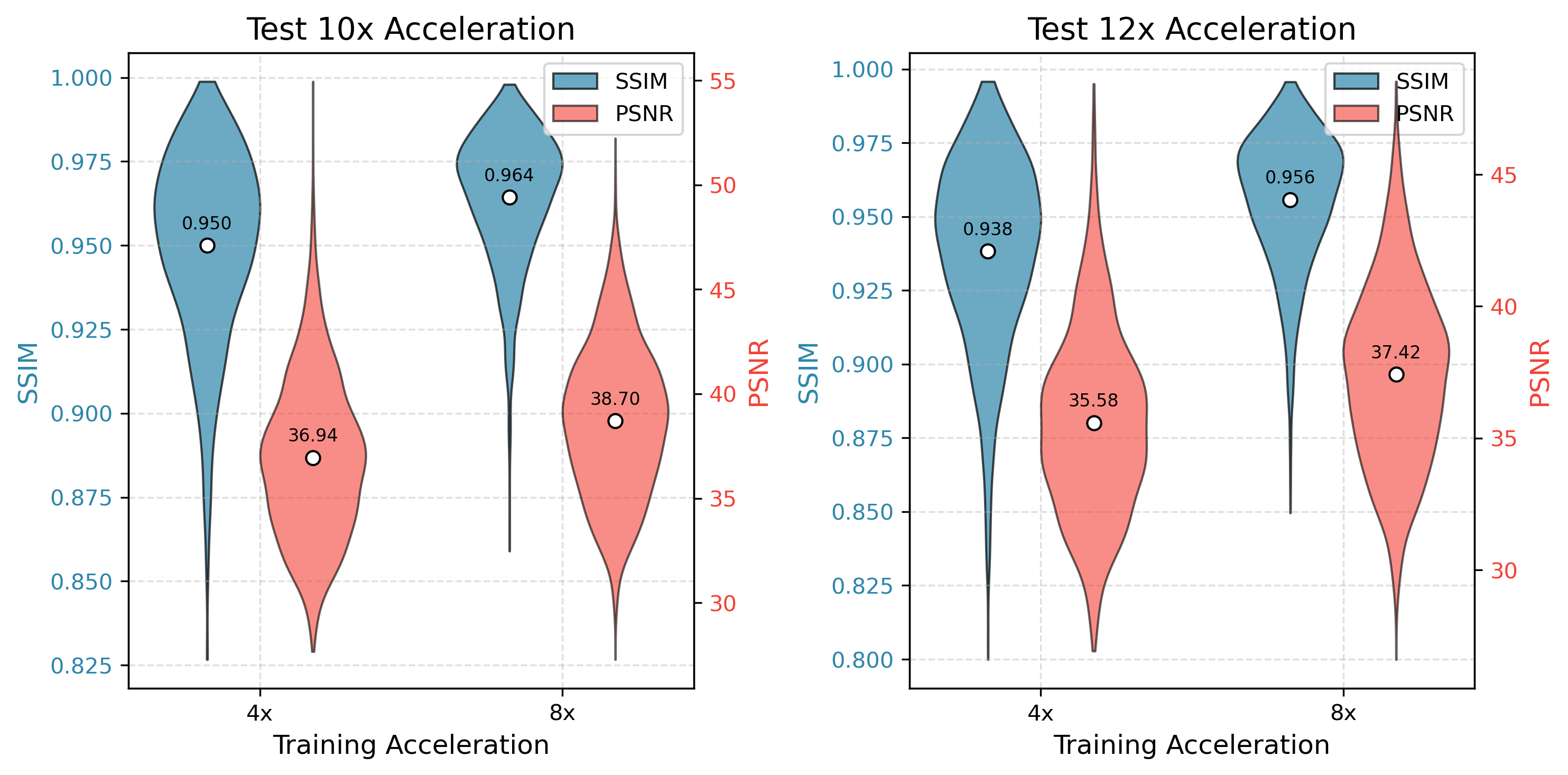}
  \caption{Generalization to extreme acceleration rates. 
  The violin plots illustrate the reconstruction performance on unseen (Left) $10\times$ and (Right) $12\times$ accelerated test data. 
  The x-axis indicates the acceleration rate used during unsupervised training ($4\times$ vs. $8\times$). 
  White circles denote the mean SSIM (blue, left axis) and PSNR (red, right axis). The results demonstrate that the model trained on higher acceleration rates ($8\times$) exhibits superior robustness and reconstruction fidelity when generalizing to extreme undersampling scenarios compared to the model trained on lower acceleration rates ($4\times$).}
  \label{fig:gen_acc_10_12}
\end{figure}

To evaluate the robustness of the proposed framework under extreme imaging conditions, we tested the models on unseen $10\times$ and $12\times$ accelerated data. 
As illustrated in \cref{fig:gen_acc_10_12}, the UPMRI framework exhibits strong generalization capabilities, with the model trained on $8\times$ data consistently outperforming the model trained on $4\times$ data across both extreme settings. 
At $10\times$ acceleration, the $8\times$-trained model achieved a mean SSIM of 0.964 and PSNR of 38.70 dB, compared to 0.950 and 36.94 dB for the $4\times$-trained model. 
Even at the highly challenging $12\times$ rate, the $8\times$-trained model maintained exceptional fidelity (SSIM 0.956, PSNR 37.42 dB), whereas the performance of the $4\times$-trained model degraded more noticeably (SSIM 0.938, PSNR 35.58 dB). 
These results suggest that unsupervised learning on sparser measurements (higher acceleration) forces the PCFM objective to learn a more robust prior, enabling it to better resolve severe aliasing artifacts in out-of-distribution testing scenarios.

\begin{figure}[t!]
  \centering
  \includegraphics[width=\linewidth]{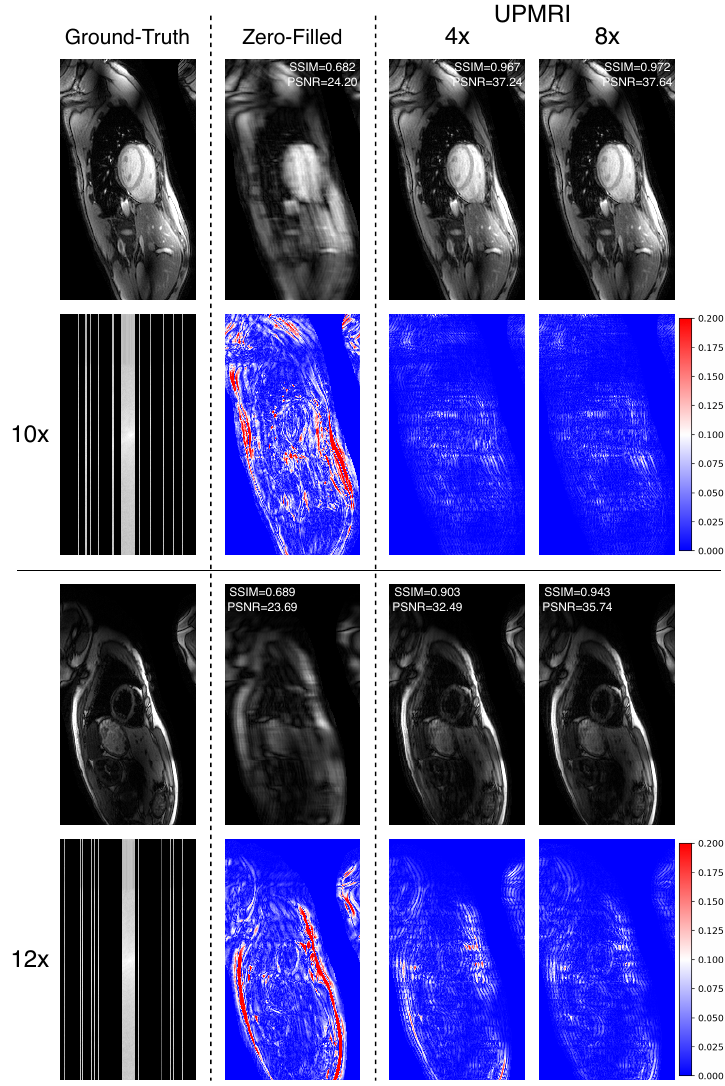}
  \caption{Qualitative visualization of generalization to extreme acceleration rates. 
  The figure displays reconstruction results on unseen test samples with (Top) $10\times$ and (Bottom) $12\times$ acceleration. 
  From left to right: the ground-truth image with the sampling mask; the zero-filled reconstruction; and the UPMRI reconstructions produced by models trained on $4\times$ and $8\times$ data, respectively. 
  The second row of each block shows the error map relative to the ground truth (blue indicates low error, red/white indicates high error). 
  The model trained on $8\times$ data demonstrates superior robustness, recovering sharper anatomical details and yielding lower error maps in these extreme undersampling scenarios compared to the model trained on $4\times$ data.}
  \label{fig:gen_vis_10_12}
\end{figure}

\cref{fig:gen_vis_10_12} presents a visual comparison of reconstructions under extreme $10\times$ and $12\times$ acceleration. 
While the zero-filled images suffer from severe blurring and heavy aliasing artifacts, the proposed UPMRI framework successfully recovers the underlying cardiac anatomy in both test cases. 
Consistent with the quantitative metrics, the model trained on 
$8\times$ undersampled data exhibits a distinct visual advantage over the model trained on $4\times$ data. 
At $12\times$ acceleration, the $4\times$-trained model begins to exhibit mild smoothing and residual artifacts; in contrast, the $8\times$-trained model preserves fine structural boundaries more effectively. 
This is evident in the error maps, where the $8\times$ column shows significantly lower error magnitudes (darker blue) than the $4\times$ column, particularly in dynamic regions like the cardiac ventricles. 
These visualizations confirm that learning from sparser measurements enables the PCFM objective to construct a prior that is more resilient to severe degradation.


%% file: 06_conclusion.tex
\section{Conclusion and Discussion}\label{sec:conclusion}
In this study, we presented UPMRI, a novel framework for unsupervised parallel MRI reconstruction that eliminates the dependency on fully sampled ground-truth data. At the core of our approach is the Projected Conditional Flow Matching (PCFM) objective, which integrates concepts from Generalized SURE to learn the prior distribution of MRI images directly from undersampled k-space measurements. By theoretically bridging the gap between the optimal PCFM solution and the marginal vector field in the measurement space, we developed a dual-space cyclic integration algorithm. This algorithm effectively leverages the learned prior while strictly enforcing data consistency through alternating forward and backward integration steps.

\subsubsection{Performance and Robustness}
Our comprehensive evaluation on the fastMRI and CMRxRecon datasets confirms that UPMRI significantly surpasses existing self-supervised and unsupervised baselines. In highly accelerated settings ($8\times$ multi-coil), UPMRI demonstrated exceptional resilience, outperforming the supervised diffusion baseline (DDNM\textsuperscript{+}) in structural similarity (SSIM) on the brain dataset and surpassing both supervised and unsupervised methods on the cardiac dataset. 

\subsubsection{Efficiency and Generalization}
Unlike standard diffusion methods that often require hundreds of function evaluations (NFEs), our flow-matching-based inference converges to high-quality solutions in as few as 20 NFEs. 
This efficiency, combined with the ability to generalize across different acceleration rates (e.g., training on $4\times$ and testing on $8\times$), highlights the practical viability of UPMRI for clinical deployment where scan protocols vary.

\subsubsection{Limitations and Future Work}
While UPMRI achieves state-of-the-art unsupervised performance, it relies on accurate estimation of coil sensitivity maps (via ESPIRiT). 
Inaccuracy in sensitivity estimation can propagate through the data consistency steps. 
Additionally, while the numerical approximation of the projection operator makes PCFM tractable for parallel MRI, it introduces a computational overhead during training compared to closed-form single-coil solutions. 
Future work will focus on extending PCFM to 3D volumetric imaging, incorporating joint sensitivity map estimation into the flow matching framework, and exploring dynamic MRI reconstruction to leverage temporal correlations.

In conclusion, UPMRI bridges the performance gap between supervised and unsupervised deep learning for MRI. By effectively learning from the data acquired in standard clinical workflows, it paves the way for high-quality, accelerated imaging without the bottleneck of acquiring fully sampled ground-truth datasets.

%% file: 07_appendix.tex
\appendices

\renewcommand{\thefigure}{A\arabic{figure}}
\renewcommand{\thetable}{A\arabic{table}}
\renewcommand{\theequation}{A\arabic{equation}}
\setcounter{figure}{0}
\setcounter{table}{0}
\setcounter{equation}{0}

\section{Proofs}
\subsection{Proof of Lemma 1}\label[appendix]{app:lemma1}

\begin{lemma}[\textbf{Measurement completeness}]
  The average projection, $\mathbb{E}_s[\bm{P}_s]$, where $\bm{P}_s=\bm{A}_s^+\bm{A}_s$, is invertible when every k-space location has a non-zero probability of being acquired.
\end{lemma}
\begin{proof}
  To prove that $\widebar{\bm{P}}\triangleq\mathbb{E}_s[\bm{P}_s]$ is invertible, we must show that its null space is trivial, meaning the only vector $\bm{x}$ for which $\widebar{\bm{P}}\bm{x}=\bm{0}$ is $\bm{x}=\bm{0}$.

  Note that the projection matrix $\bm{P}_s$ is positive semi-definite (PSD).
  The expectation of PSD matrices is also PSD.
  Therefore, $\widebar{\bm{P}}$ is also PSD.

  For a PSD matrix $\widebar{\bm{P}}$, the condition $\widebar{\bm{P}}\bm{x}=\bm{0}$ is equivalent to $\bm{x}^*\widebar{\bm{P}}\bm{x}=0$.
  Thus, we have
  \begin{equation}
    0 = \bm{x}^*\widebar{\bm{P}}\bm{x} = \bm{x}^*\mathbb{E}_s[\bm{P}_s]\bm{x} = \mathbb{E}_s[\bm{x}^*\bm{P}_s\bm{x}] = \mathbb{E}_s[\bm{x}^*\bm{P}_s^*\bm{P}_s\bm{x}] = \mathbb{E}_s\norm{\bm{P}_s\bm{x}}_2^2.
  \end{equation}
  This implies that $\bm{P}_s\bm{x}=\bm{0}$ for all $s$, which means that $\bm{x}$ must lie in the intersection of all null spaces $\mathcal{N}(\bm{P}_s)$, i.e., 
  \begin{equation}
    \bm{x} \in \bigcap_{s}\mathcal{N}(\bm{P}_s) = \bigcap_{s}\mathcal{N}(\bm{A}_s).
  \end{equation}

  This means that the Fourier transform of the coil-weighted image, $\bm{FS}_c\bm{x}$, is zero at all k-space locations sampled by mask $s$, for all $c$ and $s$.
  If every k-space location has a non-zero probability of being acquired, the union of all possible sampling patterns covers the entire k-space.
  This forces the entire Fourier-transformed image to be zero: $\bm{FS}_c\bm{x}=\bm{0}$ for all coils $c$.

  Since the Fourier transform is invertible, this means $\bm{S}_c\bm{x}=\bm{0}$ for all coils $c$, leading to $(\sum_{c}\bm{S}_c^*\bm{S}_c)\bm{x}=\bm{0}$.
  A standard assumption in parallel imaging is that there are no ``blind spots'', meaning the sum-of-squares sensitivity map matrix $\sum_{c}\bm{S}_c^*\bm{S}_c$ is invertible.
  Therefore, the only solution is $\bm{x}=\bm{0}$.
\end{proof}

\subsection{Proof of Proposition 1}\label[appendix]{app:proposition1}
\begin{proposition}[\textbf{Optimal solution to PCFM}]
  The minimizer of the PCFM objective is given by
  \begin{equation}
    \bm{v}_{\bm{\theta}^*}^X(\bm{y},t) = \mathbb{E}_{q_{t}(\bm{z}^X\mid\bm{y}),p_t^X(\bm{x}\mid\bm{z}^X)}\left[\bm{u}_t^X(\bm{x}\mid\bm{z}^X)\right],
  \end{equation}
  where $q_{t}(\bm{z}^X\mid\bm{y})=\frac{q(\bm{z}^X)p_{t}^Y(\bm{y}\mid\bm{z}^X)}{p_{t}^Y(\bm{y})}$.
  In particular, when $\bm{u}_t^X(\bm{x}\mid\bm{z}^X)=a_t'\bm{x}_0+b_t'\bm{x}_1$ is independent of $\bm{x}$, we have 
  \begin{equation}
    \bm{v}_{\bm{\theta}^*}^X(\bm{y},t) =\mathbb{E}_{q_{t}(\bm{z}^X\mid\bm{y})}\left[\bm{u}_t^X(\bm{x}\mid\bm{z}^X)\right] =\mathbb{E}\left[\bm{u}_t^X(\bm{x})\mid\bm{y}\right].
  \end{equation}
\end{proposition}

\begin{proof}
  Since $q(\bm{z}^X)p_t^X(\bm{x}\mid\bm{z}^X)p_{t}^Y(\bm{y}\mid\bm{z}^X)=p_{t}^Y(\bm{y})q_{t}^Y(\bm{z}^X\mid\bm{y})p_t^X(\bm{x}\mid\bm{z}^X)$ and by the law of total expectation, the PCFM objective can be written as 
  \begin{equation}
    \begin{aligned}
      \mathcal{L}_{\text{PCFM}}(\bm{\theta})
      &\triangleq \mathbb{E}_{s,t,q(\bm{z}^X),p_t^X(\bm{x}\mid\bm{z}^X),p_{t}^Y(\bm{y}\mid\bm{z}^X)}\norm{\bm{P}_s\!\left[\bm{v_{\theta}}^X(\bm{y},t) - \bm{u}_t^X(\bm{x}\mid\bm{z}^X)\right]}_2^2 \\
      &= \mathbb{E}_{s,t,p_{t}^Y(\bm{y})}\mathbb{E}_{q_{t}^Y(\bm{z}^X\mid\bm{y}),p_t^X(\bm{x}\mid\bm{z}^X)} \norm{\bm{P}_s\!\left[\bm{v_{\theta}}^X(\bm{y},t) - \bm{u}_t^X(\bm{x}\mid\bm{z}^X)\right]}_2^2.
    \end{aligned}
  \end{equation}
  To minimize the total expectation, we can minimize the inner conditional expectation for each value of $t$ and $\bm{y}$ independently.
  Let $\widehat{\bm{u}}_t^X(\bm{y})\triangleq\mathbb{E}_{q_{t}(\bm{z}^X\mid\bm{y}),p_t^X(\bm{x}\mid\bm{z}^X)}\left[\bm{u}_t^X(\bm{x}\mid\bm{z}^X)\right]$.
  For fixed $t$ and $\bm{y}$, the inner expectation can be transformed as 
  \begin{equation}
    \begin{aligned}
      \bm{I}_{t,\bm{y}}(\bm{\theta}) &\triangleq \mathbb{E}_{s,q_{t}^Y(\bm{z}^X\mid\bm{y}),p_t^X(\bm{x}\mid\bm{z}^X)} \norm{\bm{P}_s\!\left[\bm{v_{\theta}}^X(\bm{y},t) - \bm{u}_t^X(\bm{x}\mid\bm{z}^X)\right]}_2^2 \\
      &= \mathbb{E}_{s,q_{t}^Y(\bm{z}^X\mid\bm{y}),p_t^X(\bm{x}\mid\bm{z}^X)} \norm{\bm{P}_s\!\left[\bm{v_{\theta}}^X(\bm{y},t) - \widehat{\bm{u}}_t^X(\bm{y}) + \widehat{\bm{u}}_t^X(\bm{y}) - \bm{u}_t^X(\bm{x}\mid\bm{z}^X)\right]}_2^2 \\
      &= 
      \mathbb{E}_{s,q_{t}^Y(\bm{z}^X\mid\bm{y}),p_t^X(\bm{x}\mid\bm{z}^X)}\Big[
      \norm{\bm{P}_s\!\left[\bm{v_{\theta}}^X(\bm{y},t) - \widehat{\bm{u}}_t^X(\bm{y})\right]}_2^2 \\
      &\quad + 2\left(\bm{v_{\theta}}^X(\bm{y},t) - \widehat{\bm{u}}_t^X(\bm{y})\right)^*\bm{P}_s\left(\widehat{\bm{u}}_t^X(\bm{y}) - \bm{u}_t^X(\bm{x}\mid\bm{z}^X)\right)\Big]+\operatorname{const.},
    \end{aligned}
  \end{equation}
  where we note that
  \begin{equation}
    \begin{aligned}
      &\mathbb{E}_{q_{t}^Y(\bm{z}^X\mid\bm{y}),p_t^X(\bm{x}\mid\bm{z}^X)}\left(\bm{v_{\theta}}^X(\bm{y},t) - \widehat{\bm{u}}_t^X(\bm{y})\right)^*\bm{P}_s\left(\widehat{\bm{u}}_t^X(\bm{y}) - \bm{u}_t^X(\bm{x}\mid\bm{z}^X)\right) \\
      =\ & \left(\bm{v_{\theta}}^X(\bm{y},t) - \widehat{\bm{u}}_t^X(\bm{y})\right)^*\bm{P}_s\left[\widehat{\bm{u}}_t^X(\bm{y})-\mathbb{E}_{q_{t}^Y(\bm{z}^X\mid\bm{y}),p_t^X(\bm{x}\mid\bm{z}^X)}\bm{u}_t^X(\bm{x}\mid\bm{z}^X)\right] = 0.
    \end{aligned}
  \end{equation}
  Therefore, 
  \begin{equation}
    \begin{aligned}
      \bm{I}_{t,\bm{y}}(\bm{\theta}) &= \mathbb{E}_{s,q_{t}^Y(\bm{z}^X\mid\bm{y}),p_t^X(\bm{x}\mid\bm{z}^X)}\norm{\bm{P}_s\!\left[\bm{v_{\theta}}^X(\bm{y},t) - \widehat{\bm{u}}_t^X(\bm{y})\right]}_2^2 \\
      &= \bm{w}_{\bm{\theta}}^X(\bm{y},t)^*\mathbb{E}_s[\bm{P}_s]\bm{w}_{\bm{\theta}}^X(\bm{y},t)
    \end{aligned}
  \end{equation}
  where $\bm{w}_{\bm{\theta}}^X(\bm{y},t)\triangleq\bm{v_{\theta}}^X(\bm{y},t) - \widehat{\bm{u}}_t^X(\bm{y})$.
  By \cref{lemma1}, $\mathbb{E}_s[\bm{P}_s]$ is positive definite.
  Therefore, $\bm{I}_{t,\bm{y}}(\bm{\theta})$ is minimized when $\bm{w_{\theta}}^X(\bm{y},t)=\bm{0}$, i.e.
  \begin{equation}
    \bm{v_{\theta}}^X(\bm{y},t)=\widehat{\bm{u}}_t^X(\bm{y}).
  \end{equation}
  Furthermore, when $\bm{u}_t^X(\bm{x}\mid\bm{z}^X)=a_t'\bm{x}_0+b_t'\bm{x}_1$ is independent of $\bm{x}$, due to the Markov chain $\bm{z}^X\rightarrow\bm{x}\rightarrow\bm{y}$ and the law of iterated expectation, we have
  \begin{equation}
    \begin{aligned}
      \widehat{\bm{u}}_t^X(\bm{y}) &= \mathbb{E}\left[a_t'\bm{x}_0+b_t'\bm{x}_1\mid\bm{y}\right] \\
      &= \mathbb{E}\left[\mathbb{E}\left[a_t'\bm{x}_0+b_t'\bm{x}_1\mid\bm{x}\right]\mid\bm{y}\right] \\
      &= \mathbb{E}\left[\bm{u}_t^X(\bm{x})\mid\bm{y}\right].
    \end{aligned}
  \end{equation}
\end{proof}

\subsection{Proof of Proposition 2}\label[appendix]{app:proposition2}
\begin{proposition}[\textbf{Unsupervised transformation of PCFM}]
  Assuming deterministic conditional probability paths $p_t^X(\bm{x}\mid\bm{z}^X)=\delta_{a_t\bm{x}_0+b_t\bm{x}_1}(\bm{x})$ and $p_t^Y(\bm{y}\mid\bm{z}^Y)=\delta_{a_t\bm{y}_0+b_t\bm{y}_1}(\bm{y})$ with $\bm{y}_{0}=\bm{A}_s\bm{x}_{0}+\bm{e}_{0}$ and $\bm{y}_{1}=\bm{A}_s\bm{x}_{1}$, where $\bm{e}_{0}\sim\mathcal{CN}(\bm{0},\sigma_0^2\bm{I}_{Cd})$, then up to a constant the PCFM objective can be transformed to
  \begin{equation}
    \begin{aligned}
      \mathbb{E}_{s,t,q(\bm{z}^Y),p_{t}^Y(\bm{y}\mid\bm{z}^Y)}
    \Big[
      \norm{\bm{P}_s\!\left[\bm{v_{\theta}}^X(\bm{A}_s^*\bm{y},t)-\widehat{\bm{u}}_{s,t,\text{ML}}^X\right]}_2^2
      + {2a_t}{a_t'}\sigma_0^2\nabla_{\bm{A}_s^*\bm{y}}\cdot\bm{P}_s\bm{v_{\theta}}^X(\bm{A}_s^*\bm{y},t)
    \Big],
    \end{aligned}
  \end{equation}
  where $q(\bm{z}^Y)=q(\bm{y}_0)q(\bm{y}_1)=q(\bm{y}_0)\mathbb{E}_{q(\bm{x}_1)} \left[p(\bm{y}_1\mid\bm{x}_1)\right]$ is sampled by the MRI forward model and Monte Carlo estimation, $\bm{P}_s=\bm{A}_s^+\bm{A}_s$ is the range-space projection, and 
  \begin{equation}
    \widehat{\bm{u}}_{s,t,\text{ML}}^X\triangleq (\bm{A}_s^*\bm{C}_t^{-1}\bm{A}_s)^+\bm{A}_s^*\bm{C}_t^{-1}\bm{u}_{t}^Y(\bm{y}\mid\bm{z}^Y)
  \end{equation}
  with $\bm{C}_t=(a_t'\sigma_0)^2\bm{I}_d$ is the maximum likelihood solution of the forward model in 
  \begin{equation}
    \bm{u}_{t}^Y(\bm{y}\mid\bm{z}^Y) =\bm{A}_s\bm{u}_{t}^X(\bm{x}\mid\bm{z}^X) + a_t'\bm{e}_{0}.
  \end{equation}
  Note that range-space projection can be approximated by the conjugate gradient method.
\end{proposition}
\begin{proof}
  The proof is inspired from \citep{journal/tsp/eldar2008}.
  We omit the subscript $s$ without loss of generality.
  Noting the deterministic mapping between the $\mathcal{Y}$-space conditional path and the conditional vector field
  \begin{equation}\label{eq:mapping}
    \bm{y} = \frac{a_t}{a_t'}\bm{u}_{t}^Y(\bm{y}\mid\bm{z}^Y) - b_t'\left(\frac{a_t}{a_t'}-\frac{b_t}{b_t'}\right)\bm{y}_{1},
  \end{equation}
  we can write $\bm{v_{\theta}}^X(\bm{y},t)=\bm{v_{\theta}}^X\!\left(\bm{u}_{t}^Y(\bm{y}\mid\bm{z}^Y),t\right)$ and the objective as
  \begin{equation}
    \mathcal{L}_{\text{PCFM}}(\bm{\theta}) = \mathbb{E}_{t,q(\bm{z}^X),p_t^X(\bm{x}\mid\bm{z}^X),p_{t}^Y(\bm{y}\mid\bm{z}^X)}\norm{\bm{P}\!\left[\bm{v_{\theta}}^X\!\left(\bm{u}_{t}^Y(\bm{y}\mid\bm{z}^Y),t\right)-\bm{u}_{t}^X(\bm{x}\mid\bm{z}^X)\right]}_2^2.
  \end{equation}
  By the linear forward model over the dual-space conditional vector fields
  \begin{equation}
    \bm{u}_{t}^Y(\bm{y}\mid\bm{z}^Y) = \bm{A}\bm{u}_{t}^X(\bm{x}\mid\bm{z}^X) + a_t'\bm{e}_{0},
  \end{equation}
  we note that the sufficient statistic $\bm{\mu}_{t}^X\triangleq\bm{A}^*\bm{C}_t^{-1}\bm{u}_{t}^Y$ follows a Gaussian distribution $\mathcal{CN}(\bm{A}^*\bm{C}_t^{-1}\bm{A}\bm{u}_t^X,\bm{A}^*\bm{C}_t^{-1}\bm{A})$ with probability density function (pdf)
  \begin{equation}
    p(\bm{\mu}_{t}^X) = q(\bm{\mu}_{t}^X)\exp\left({\bm{u}_t^X}^*\bm{\mu}_{t}^X - g(\bm{u}_t^X)\right),
  \end{equation}
  where
  \begin{equation}
    \begin{aligned}
      q(\bm{\mu}_{t}^X) &= K\cdot \exp\left( -\frac{1}{2}{\bm{\mu}_{t}^X}^*\left(\bm{A}^*\bm{C}_t^{-1}\bm{A}\right)^+ \bm{\mu}_{t}^X\right), \\
      g(\bm{u}_t^X) &= \frac{1}{2}{\bm{u}_t^X}^*\bm{A}^*\bm{C}_t^{-1}\bm{A}\bm{u}_t^X.
    \end{aligned}
  \end{equation}
  Assuming the network's input to be $\bm{\mu}_{t}^X$, we can write
  \begin{equation}
    \begin{aligned}
      \mathcal{L}_{\text{PCFM}}(\bm{\theta}) = \mathbb{E}_{t,q(\bm{z}^X),p_t^X(\bm{x}\mid\bm{z}^X),p_{t}^Y(\bm{y}\mid\bm{z}^X)} 
    \left[
      {\bm{u}_t^X}^*\bm{P}\bm{u}_t^X + {\bm{v_{\theta}}^X(\bm{\mu}_{t}^X,t)}^*\bm{P}\bm{v_{\theta}}^X(\bm{\mu}_{t}^X,t) - 2{\bm{u}_t^X}^*\bm{P} \bm{v_{\theta}}^X(\bm{\mu}_{t}^X,t)
    \right],
    \end{aligned}
  \end{equation}
  and
  \begin{equation}
    \begin{aligned}
      \mathbb{E}_{p_{t}^Y(\bm{y}\mid\bm{z}^X)}\left[{\bm{u}_t^X}^*\bm{P} \bm{v_{\theta}}^X(\bm{\mu}_{t}^X,t)\right] &= \mathbb{E}_{p(\bm{\mu}_{t}^X)}\left[{\bm{u}_t^X}^*\bm{P} \bm{v_{\theta}}^X(\bm{\mu}_{t}^X,t)\right] \\
      &= \int {\bm{v_{\theta}}^X(\bm{\mu}_{t}^X,t)}^*\bm{P}\bm{u}_t^X \cdot q(\bm{\mu}_{t}^X)\exp\left({\bm{u}_t^X}^*\bm{\mu}_{t}^X - g(\bm{u}_t^X)\right) \dd\bm{\mu}_{t}^X.
    \end{aligned}
  \end{equation}
  Denote $h(\bm{\mu}_{t}^X)\triangleq \exp\left({\bm{u}_t^X}^*\bm{\mu}_{t}^X - g(\bm{u}_t^X)\right)$.
  Substituting $\bm{u}_t^Xh(\bm{\mu}_{t}^X)=\nabla_{\bm{\mu}_{t}^X}h(\bm{\mu}_{t}^X)$ and integrating by parts, we have
  \begin{equation}
    \begin{aligned}
      \mathbb{E}_{p(\bm{\mu}_{t}^X)}\left[{\bm{u}_t^X}^*\bm{P} \bm{v_{\theta}}^X(\bm{\mu}_{t}^X,t)\right] &= \int {\bm{v_{\theta}}^X(\bm{\mu}_{t}^X,t)}^*\bm{P}\bm{u}_t^X \cdot q(\bm{\mu}_{t}^X)\nabla_{\bm{\mu}_{t}^X}h(\bm{\mu}_{t}^X) \dd\bm{\mu}_{t}^X \\
      &= - \int h(\bm{\mu}_{t}^X) \nabla_{\bm{\mu}_{t}^X}\cdot
      \left[
        q(\bm{\mu}_{t}^X)\bm{P}\bm{v_{\theta}}^X(\bm{\mu}_{t}^X,t)
      \right] \dd\bm{\mu}_{t}^X,
    \end{aligned}
  \end{equation}
  where 
  \begin{equation}
    \nabla_{\bm{\mu}_{t}^X}\cdot\left[q(\bm{\mu}_{t}^X)\bm{P}\bm{v_{\theta}}^X(\bm{\mu}_{t}^X,t)\right] = 
    q(\bm{\mu}_{t}^X)\left[\nabla_{\bm{\mu}_{t}^X}\cdot \bm{P}\bm{v_{\theta}}^X(\bm{\mu}_{t}^X,t) + {\bm{v_{\theta}}^X(\bm{\mu}_{t}^X,t)}^*\bm{P}\nabla_{\bm{\mu}_{t}^X}\ln q(\bm{\mu}_{t}^X)\right]
  \end{equation}
  and $\ln q(\bm{\mu}_{t}^X)= - \left(\bm{A}^*\bm{C}_t^{-1}\bm{A}\right)^+\bm{\mu}_{t}^X=-\widehat{\bm{u}}_{t,\text{ML}}^X$.
  Therefore, 
  \begin{equation}
    \begin{aligned}
      \mathbb{E}_{p(\bm{\mu}_{t}^X)}\left[{\bm{u}_t^X}^*\bm{P} \bm{v_{\theta}}^X(\bm{\mu}_{t}^X,t)\right] &=
      \mathbb{E}_{p(\bm{\mu}_{t}^X)}\left[-\nabla_{\bm{\mu}_{t}^X}\cdot \bm{P}\bm{v_{\theta}}^X(\bm{\mu}_{t}^X,t)+{\bm{v_{\theta}}^X(\bm{\mu}_{t}^X,t)}^*\bm{P}\widehat{\bm{u}}_{t,\text{ML}}^X \right]
    \end{aligned}
  \end{equation}
  where
  \begin{equation}
    \begin{aligned}
      \mathcal{L}_{\text{PCFM}}(\bm{\theta}) 
      &= \mathbb{E}\left[
        {\bm{u}_t^X}^*\bm{P}\bm{u}_t^X + {\bm{v_{\theta}}^X(\bm{\mu}_{t}^X,t)}^*\bm{P}\bm{v_{\theta}}^X(\bm{\mu}_{t}^X,t) + 2\nabla_{\bm{\mu}_{t}^X}\cdot \bm{P}\bm{v_{\theta}}^X(\bm{\mu}_{t}^X,t) - 2 {\bm{v_{\theta}}^X(\bm{\mu}_{t}^X,t)}^*\bm{P}\widehat{\bm{u}}_{t,\text{ML}}^X\right] \\
        &= \mathbb{E}\left[
          \norm{\bm{P}\!\left[ \bm{v_{\theta}}^X(\bm{\mu}_{t}^X,t) - \widehat{\bm{u}}_{t,\text{ML}}^X \right]}_2^2 
          + 2 \nabla_{\bm{\mu}_{t}^X}\cdot \bm{P}\bm{v_{\theta}}^X(\bm{\mu}_{t}^X,t)
          + \norm{\bm{P}\bm{u}_t^X}_2^2 - \norm{\bm{P}\widehat{\bm{u}}_{t,\text{ML}}^X}_2^2
        \right] \\
        &= \mathbb{E}_{t,q(\bm{z}^X),p_t^X(\bm{x}\mid\bm{z}^X),p_{t}^Y(\bm{y}\mid\bm{z}^X)}\left[
          \norm{\bm{P}\!\left[ \bm{v_{\theta}}^X(\bm{\mu}_{t}^X,t) - \widehat{\bm{u}}_{t,\text{ML}}^X \right]}_2^2 
          + 2 \nabla_{\bm{\mu}_{t}^X}\cdot \bm{P}\bm{v_{\theta}}^X(\bm{\mu}_{t}^X,t)
        \right] +\operatorname{const.}
    \end{aligned}
  \end{equation}
  Then, using \cref{eq:mapping} and writing back $\bm{v_{\theta}}^X(\bm{\mu}_{t}^X,t)=\bm{v_{\theta}}^X(\bm{A}^*\bm{y},t)$, we obtain by change of variables that $\mathcal{L}_{\text{PCFM}}(\bm{\theta})$ can be transformed to the following expression up to a constant
  \begin{equation}
    \mathbb{E}_{t,q(\bm{z}^X),p_{t}^Y(\bm{y}\mid\bm{z}^X)}
    \left[
      \norm{\bm{P}\!\left[\bm{v_{\theta}}^X(\bm{A}^*\bm{y},t)-\widehat{\bm{u}}_{t,\text{ML}}^X\right]}_2^2
      + {2a_t}{a_t'}\sigma_0^2\nabla_{\bm{A}^*\bm{y}}\cdot\bm{P}\bm{v_{\theta}}^X(\bm{A}^*\bm{y},t)
    \right].
  \end{equation}
  Finally, it concludes the proof by noting that 
  \begin{equation}
    \mathbb{E}_{q(\bm{z}^X)}\left[ p_t^Y(\bm{y}\mid\bm{z}^X) \right] = p_t^Y(\bm{y}) = \mathbb{E}_{q(\bm{z}^Y)}\left[ p_t^Y(\bm{y}\mid\bm{z}^Y) \right],
  \end{equation}
  where $q(\bm{z}^Y)=q(\bm{y}_0)\mathbb{E}_{q(\bm{x}_1)} \left[p(\bm{y}_1\mid\bm{x}_1)\right]$ is sampled by the MRI and Monte Carlo estimation.
\end{proof}

\subsection{Proof of Lemma 2}\label[appendix]{app:lemma2}
\begin{lemma}
  The $\mathcal{Y}$-space marginal vector field that generates the probability path $p_t^Y$ takes the form
  \begin{equation}
    \bm{u}_{t}^Y(\bm{y}) = \bm{Av}_{\bm{\theta}^*}^X(\bm{y},t) - a_ta_t'\sigma_0^2\nabla_{\bm{y}}\log p_t^Y(\bm{y}).
  \end{equation}
\end{lemma}

\begin{proof}
  The $\mathcal{Y}$-space conditional vector field is given by
  \begin{equation}
    \bm{u}_t^Y(\bm{y}\mid\bm{z}^X) = \bm{Au}_t^X(\bm{x}\mid\bm{z}^X)-a_ta_t'\sigma_0^2\nabla_{\bm{y}}\log p_{t}^Y(\bm{y}\mid\bm{z}^X).
  \end{equation}
  Therefore, by \cref{proposition1}, the $\mathcal{Y}$-space marginal vector field takes the form
  \begin{equation}
    \begin{aligned}
      \bm{u}_t^Y(\bm{y}) &= \mathbb{E}_{q_t(\bm{z}^X\mid\bm{y})}\left[ \bm{u}_t^Y(\bm{y}\mid\bm{z}^X) \right] \\
      &= \bm{A}\mathbb{E}_{q_t(\bm{z}^X\mid\bm{y})}\left[\bm{u}_t^X(\bm{x}\mid\bm{z}^X)\right] - a_ta_t'\sigma_0^2\mathbb{E}_{q_t(\bm{z}^X\mid\bm{y})}\left[\nabla_{\bm{y}}\log p_{t}^Y(\bm{y}\mid\bm{z}^X)\right] \\
      &= \bm{A}\bm{v}_{\bm{\theta}^*}^X(\bm{y},t) - a_ta_t'\sigma_0^2\nabla_{\bm{y}}\log p_t^Y(\bm{y}),
    \end{aligned}
  \end{equation}
  where we have used the fact that
  \begin{equation}
    \begin{aligned}
      \mathbb{E}_{q_t(\bm{z}^X\mid\bm{y})}\left[\nabla_{\bm{y}}\log p_{t}^Y(\bm{y}\mid\bm{z}^X)\right] &= \int q_t(\bm{z}^X\mid\bm{y})\frac{1}{p_{t}^Y(\bm{y}\mid\bm{z}^X)} \nabla_{\bm{y}} p_{t}^Y(\bm{y}\mid\bm{z}^X)\dd\bm{z}^X \\
      &= \frac{1}{p_t^Y(\bm{y})}\int q(\bm{z}^X)\nabla_{\bm{y}} p_{t}^Y(\bm{y}\mid\bm{z}^X)\dd\bm{z}^X \\
      &= \frac{1}{p_t^Y(\bm{y})}\nabla_{\bm{y}} \int q(\bm{z}^X) p_{t}^Y(\bm{y}\mid\bm{z}^X)\dd\bm{z}^X \\
      &= \frac{1}{p_t^Y(\bm{y})}\nabla_{\bm{y}} p_t^Y(\bm{y}) \\
      &= \nabla_{\bm{y}}\log p_t^Y(\bm{y}).
    \end{aligned}
  \end{equation}
\end{proof}

\subsection{Proof of Lemma 3}\label[appendix]{app:lemma3}
\begin{lemma}
  Note that $p_1^Y(\bm{y})= \int p_1(\bm{y}\mid\bm{x})p_1^X(\bm{x})\dd\bm{x}= \mathcal{CN}(\bm{y}\mid\bm{0},2\bm{AA}^*)$.
  Then, 
  \begin{equation}
    \bm{u}_t^Y(\bm{y}) = \frac{a_t'}{a_t}\bm{y}-b_t\left(b_t'-\frac{a_t'}{a_t}b_t\right)(2\bm{AA}^*)\nabla_{\bm{y}}\log p_t^Y(\bm{y}).
  \end{equation}
\end{lemma}

\begin{proof}
  Taking $\bm{y}_0$ as the conditioning variable, we have 
  \begin{equation}\label{eq:lemma2_1}
    \begin{aligned}
      \bm{u}_t^Y(\bm{y}) &= \mathbb{E}_{q_t(\bm{y}_0\mid\bm{y})}\left[ a_t'\bm{y}_0 + b_t'\bm{y}_1 \right] \\
      &= \mathbb{E}_{q_t(\bm{y}_0\mid\bm{y})} \left[ a_t'\frac{\bm{y}-b_t\bm{y}_1}{a_t} + b_t'\bm{y}_1 \right] \\ 
      &= \frac{a_t'}{a_t}\bm{y} + \left(b_t'-\frac{a_t'}{a_t}b_t\right)\mathbb{E}_{q_t(\bm{y}_0\mid\bm{y})}[\bm{y}_1],
    \end{aligned}
  \end{equation}
  where $\bm{y}_1=\frac{\bm{y}-a_t\bm{y}_0}{b_t}$.
  
  On the other hand, 
  noting that $p_t^Y(\bm{y}\mid\bm{y}_0)=\mathcal{CN}\left(\bm{y}\mid a_t\bm{y}_0, b_t^2(2\bm{AA}^*)\right)$,
  the score function can be written as 
  \begin{equation}\label{eq:lemma2_2}
    \begin{aligned}
      \nabla_{\bm{y}}\log p_t^Y(\bm{y}) &= \frac{1}{p_t^Y(\bm{y})}\nabla_{\bm{y}}p_t^Y(\bm{y}) \\
      &= \frac{1}{p_t^Y(\bm{y})}\nabla_{\bm{y}}\int p_t^Y(\bm{y}\mid\bm{y}_0)q(\bm{y}_0)\dd\bm{y}_0 \\
      &= \frac{1}{p_t^Y(\bm{y})}\int p_t^Y(\bm{y}\mid\bm{y}_0)q(\bm{y}_0)\nabla_{\bm{y}}\log p_t^Y(\bm{y}\mid\bm{y}_0)\dd\bm{y}_0 \\
      &= -\frac{1}{p_t^Y(\bm{y})}\int p_t^Y(\bm{y}\mid\bm{y}_0)q(\bm{y}_0) \frac{(2\bm{AA}^*)^{+}(\bm{y}-a_t\bm{y}_0)}{b_t^2}\dd\bm{y}_0 \\
      &= - \int q_t(\bm{y}_0\mid\bm{y})\frac{(2\bm{AA}^*)^{+}b_t\bm{y}_1}{b_t^2}\dd\bm{y}_0 \\
      &= - \frac{(2\bm{AA}^*)^{+}}{b_t} \mathbb{E}_{q_t(\bm{y}_0\mid\bm{y})}[\bm{y}_1].
    \end{aligned}
  \end{equation}
  Combining \cref{eq:lemma2_1} and \cref{eq:lemma2_2} by canceling out $\mathbb{E}_{q_t(\bm{y}_0\mid\bm{y})}[\bm{y}_1]$ completes the proof.
\end{proof}

\subsection{Proof of Proposition 3}\label[appendix]{app:proposition3}
\begin{proposition}[\textbf{Vector fields under projection}]
  For $a_t=1-t$ and $b_t=t$, the $\mathcal{Y}$-space marginal vector field $\bm{u}_t^Y(\bm{y})$ can be expressed by $\bm{v}_{\bm{\theta}^*}^X(\bm{y},t)$ as
  \begin{equation}
    \bm{u}_t^Y(\bm{y}) = \bm{Av}_{\bm{\theta}^*}^X(\bm{y},t) - \frac{c_t}{1-t}\left[c_t\bm{I}_{Cd}+2\bm{AA}^*\right]^{-1} \left[(1-t)\bm{Av}_{\bm{\theta}^*}^X(\bm{y},t)+\bm{y}\right],
  \end{equation}
  where $c_t\triangleq \frac{(1-t)^2}{t}\sigma_0^2$.
  In addition, left-multiplying both sides with $\bm{A}^*$ gives the more computationally friendly formula when $Cd>D$:
  \begin{equation}
    \bm{A}^*\bm{u}_t^Y(\bm{y}) = \bm{A}^*\bm{Av}_{\bm{\theta}^*}^X(\bm{y},t) - \frac{c_t}{1-t}\left[c_t\bm{I}_{D}+2\bm{A}^*\bm{A}\right]^{-1} \bm{A}^*\left[(1-t)\bm{Av}_{\bm{\theta}^*}^X(\bm{y},t)+\bm{y}\right].
  \end{equation}
\end{proposition}
\begin{proof}
  For $a_t=1-t$ and $b_t=t$, \cref{lemma3} indicates
  \begin{equation}
    \nabla_{\bm{y}}\log p_t^Y(\bm{y}) = -\frac{1}{t}(2\bm{AA}^*)^{+} \left[ \bm{y} + (1-t)\bm{u}_t^Y(\bm{y)} \right].
  \end{equation}
  Substituting this into \cref{lemma2} gives the equation
  \begin{equation}
    \bm{u}_t^Y(\bm{y}) = \bm{Av}_{\bm{\theta}^*}^X(\bm{y},t) -\left( \frac{1-t}{t}\sigma_0^2\right)(2\bm{AA}^*)^{+}\left[\bm{y}+(1-t)\bm{u}_t^Y(\bm{y})\right].
  \end{equation}
  Denoting $\bm{u}\triangleq \bm{u}_t^Y(\bm{y})$ and $\bm{v}\triangleq \bm{v}_{\bm{\theta}^*}^X(\bm{y},t)$, then solving for $\bm{u}$ in the above equation gives
  \begin{equation}
    \bm{u} = \left(\bm{I}+c(2\bm{AA}^*)^{+}\right)^{-1}\left(\bm{Av}-\frac{c}{1-t}(2\bm{AA}^*)^{+}\bm{y}\right),
  \end{equation}
  where $c\triangleq \frac{(1-t)^2}{t}\sigma_0^2$.
  
  Note that
  \begin{equation}
    \left(\bm{I}+c(2\bm{AA}^*)^{+}\right)^{-1} = \bm{I}-c\bm{AA}^+\left(c\bm{I}+2\bm{AA}^*\right)^{-1},
  \end{equation}
  and 
  \begin{equation}
    \left(\bm{I}+c(2\bm{AA}^*)^{+})\right)^{-1}(2\bm{AA}^*)^{+}=\bm{AA}^+\left(c\bm{I}+2\bm{AA}^*\right)^{-1}.
  \end{equation}
  Therefore, 
  \begin{equation}
    \begin{aligned}
      \bm{u} &= \left[\bm{I}-c\bm{AA}^+\left(c\bm{I}+2\bm{AA}^*\right)^{-1}\right]\bm{Av} - \frac{c}{1-t} \bm{AA}^+\left(c\bm{I}+2\bm{AA}^*\right)^{-1}\bm{y} \\
      &= \bm{Av} - c\bm{AA}^+\left(c\bm{I}+2\bm{AA}^*\right)^{-1}\left[\bm{Av}+\frac{1}{1-t}\bm{y}\right] \\
      &= \bm{Av} - \frac{c}{1-t}\bm{AA}^+\left(c\bm{I}+2\bm{AA}^*\right)^{-1}\left[(1-t)\bm{Av}+\bm{y}\right],
    \end{aligned}
  \end{equation}
  and by the identity $\bm{A}^*\bm{AA}^+(c\bm{I}_{Cd}+2\bm{AA}^*)^{-1}=(c\bm{I}_D+2\bm{A}^*\bm{A})^{-1}\bm{A}^*$,
  \begin{equation}
    \bm{A}^*\bm{u} = \bm{A}^*\bm{Av} - \frac{c}{1-t}(c\bm{I}_D+2\bm{A}^*\bm{A})^{-1}\bm{A}^*\left[(1-t)\bm{Av}+\bm{y}\right].
  \end{equation}
\end{proof}

\subsection{Proof of Lemma 4}\label[appendix]{app:lemma4}
\begin{lemma}
  Assuming that $p(\bm{x}_1)=\mathcal{CN}(\bm{0}, 2\bm{I}_D)$ and $\bm{y}_1=\bm{Ax}_1$, then we have
  \begin{equation}
    p(\bm{x}_1\mid\bm{y}_1) = \mathcal{CN}\left(\bm{x}_1\mid\bm{A}^+\bm{y}_1, 2(\bm{I}_D-\bm{A}^+\bm{A})\right).
  \end{equation}
\end{lemma}
\begin{proof}
  We need to calculate the covariance matrices $\bm{\Sigma}_{xx}$, $\bm{\Sigma}_{xy}$, $\bm{\Sigma}_{yx}$, $\bm{\Sigma}_{yy}$.
  The prior covariance is given by 
  \begin{equation}
    \bm{\Sigma}_{xx} = 2\bm{I}_D.
  \end{equation}
  The cross covariance is given by 
  \begin{equation}
    \begin{aligned}
     \bm{\Sigma}_{xy} &= \mathbb{E}[\bm{xy}^*] = \mathbb{E}[\bm{xx}^*]\bm{A}^* =2\bm{A}^*, \\
     \bm{\Sigma}_{yx}&= \bm{\Sigma}_{xy}^* = 2\bm{A}.
    \end{aligned}
  \end{equation}
  The observation covariance is given by 
  \begin{equation}
    \bm{\Sigma}_{yy} = \mathbb{E}[\bm{yy}^*] = \bm{A}\mathbb{E}[\bm{xx}^*]\bm{A}^* = 2\bm{AA}^*.
  \end{equation}
  Thus, the joint distribution is 
  \begin{equation}
    \mqty[\bm{x}_1\\\bm{y}_1] \sim \mathcal{CN}\left(\mqty[\bm{0}\\\bm{0}], 2\mqty[\bm{I} & \bm{A}^* \\ \bm{A} & \bm{AA}^*]\right).
  \end{equation}
  The conditional mean and covariance of the distribution $p(\bm{x}_1\mid\bm{y}_1)$ is given by
  \begin{equation}
    \begin{aligned}
      \bm{\mu}_{x\mid y} &= \bm{\mu}_x +\bm{\Sigma}_{xy}\bm{\Sigma}_{yy}^{-1}(\bm{y}_1-\bm{\mu}_y) = \bm{A}^+\bm{y}_1, \\
      \bm{\Sigma}_{x\mid y} &= \bm{\Sigma}_{xx} - \bm{\Sigma}_{xy}\bm{\Sigma}_{yy}^{-1}\bm{\Sigma}_{yx} = 2\bm{I}_D - (2\bm{A}^*)(2\bm{AA}^*)^{-1}(2\bm{A}) = 2(\bm{I}_D-\bm{A}^+\bm{A}).
    \end{aligned}
  \end{equation}
\end{proof}

\subsection{Proof of Proposition 4}\label[appendix]{app:proposition4}
\begin{proposition}[\textbf{Measurement consistency}]
  Given the sample 
  \begin{equation}
    \bm{x}_1 = \bm{A}^+\bm{y}_1 + (\bm{I}_D-\bm{A}^+\bm{A})\bm{z}_1,\quad \bm{z}_1\sim\mathcal{CN}(\bm{0},2\bm{I}_D),
  \end{equation}
  the forward process
  \begin{equation}
    \frac{\dd\bm{y}(t)}{\dd t} = \bm{u}_t^Y(\bm{y}).
  \end{equation}
  and the backward process 
  \begin{equation}
    \frac{\dd\bm{x}(t)}{\dd t} = \bm{v}_{\bm{\theta}^*}^X(\bm{y}(t),t),
  \end{equation}
  then as $\sigma_0\rightarrow 0$, we have
  \begin{equation}
    \bm{Ax}(t) = \bm{y}(t), \quad\forall\, t\in [0,1]
  \end{equation}
  which states that the backward process is consistent with the forward process.
  In particular, we have $\bm{Ax}_0=\bm{y}_0$.
\end{proposition}
\begin{proof}
  According to \cref{proposition1}, we have 
  \begin{equation}
    \begin{aligned}
      \bm{A}\bm{v}_{\bm{\theta}^*}^X(\bm{y},t) &= \bm{A}\mathbb{E}\left[a_t'\bm{x}_0+b_t'\bm{x}_1\mid\bm{y}\right] \\ 
      &=\mathbb{E}\left[a_t'\bm{y}_0+b_t'\bm{y}_1-a_t'\bm{e}\mid\bm{y}\right] \\
      &=\mathbb{E}\left[a_t'\bm{y}_0+b_t'\bm{y}_1\mid\bm{y}\right] - \mathbb{E}[a_t'\bm{e}\mid\bm{y}] \\
      &\rightarrow \bm{u}_t^Y(\bm{y}),
    \end{aligned}
  \end{equation}
  as $\sigma_0\rightarrow 0$.
  Therefore, we have
  \begin{equation}
    \begin{aligned}
      \bm{A}\bm{x}(t) = \bm{A}\bm{x}_1 + \bm{A}\int_1^t\bm{v}_{\bm{\theta}^*}^X(\bm{y}({\tau}),\tau)\dd\tau \rightarrow \bm{y}_1 + \int_1^t \bm{u}_{\tau}^Y(\bm{y})\dd\tau = \bm{y}(t).
    \end{aligned}
  \end{equation}
\end{proof}

\section{Numerical Methods}\label[appendix]{app:numeric}
\subsection{Conjugate gradient}\label[appendix]{app:cg}
Conjugate gradient (CG) is a method for solving a linear system of equations $\bm{Ax}=\bm{b}$, when the matrix $\bm{A}$ is symmetric (or Hermitian) positive-definite (SPD) and very large (often sparse).
Direct methods like Gaussian elimination are impractical due to computational cost and memory requirements.
CG reframes the problem of solving a linear system as an optimization problem.
The solution to $\bm{Ax}=\bm{b}$ is precisely the vector $\bm{x}$ that minimizes the quadratic form:
\begin{equation}
  \phi(\bm{x}) = \frac{1}{2}\bm{x}^*\bm{Ax} - \bm{x}^*\bm{b}.
\end{equation}
An intuitive way to find this minimum is to take the steepest descent, where one repeatedly steps in the direction of the negative gradient $-\nabla\phi(\bm{x})=\bm{b}-\bm{Ax}$.
However, the steepest descent often performs poorly, taking many small zigzagging steps to reach the minimum.

CG dramatically improves upon this by choosing a sequence of search directions that are ``smarter''.
Instead of using the residual at each step, it generates a set of search directions $\{\bm{p}_0,\bm{p}_1,\dots,\bm{p}_{k-1}\}$ that are mutually $\bm{A}$-orthogonal, i.e., $\bm{p}_i^*\bm{A}\bm{p}_j=0$ for $i\neq j$.
This property is crucial: minimizing the quadratic function along a new search direction $\bm{p}_k$ does not compromise the minimization that has already been achieved in the previous directions.
At each iteration $k$, CG finds the optimal solution $\bm{x}_k$ within the affine Krylov subspace $\bm{x}_0+\mathcal{K}_k(\bm{A},\bm{r}_0)$, where $\bm{r}_0\triangleq\bm{b}-\bm{Ax}_0$ is the initial residual.
This means that after $k$ steps, CG has found the best possible solution that can be formed by a linear combination of $\{\bm{r}_0,\bm{Ar}_0,\dots,\bm{A}^{k-1}\bm{r}_0\}$.
This guarantees convergence for an $N\times N$ matrix in at most $N$ steps, though in practice, a good approximation is often found in far fewer iterations.
\cref{alg:cg} provides the CG algorithm in detail.

\begin{algorithm}[t]
  \caption{Conjugate Gradient (CG)}
  \label{alg:cg}
  \KwIn{A symmetric (or Hermitian) matrix $\bm{A}$, a vector $\bm{b}$, an intial guess $\bm{x}_0$, a maximum number of iterations $k$, and a tolerance $\epsilon$.}
  \KwOut{An approximate solution $\bm{x}$ to $\bm{Ax}=\bm{b}$.}
  Set $\bm{r}_0\gets\bm{b}-\bm{Ax}_0$ and $\bm{p}_0\gets\bm{r}_0$. \algorithmiccomment{Initialization} \\
  \For{$j=1,\dots,k$}{
    $\bm{v}_j\gets\bm{Ap}_j$ \\
    $\alpha_j \gets \frac{\bm{r}_j^*\bm{r}_j}{\bm{p}_j^*\bm{v}_j}$ \algorithmiccomment{Compute the step size} \\
    $\bm{x}_{j+1}\gets\bm{x}_j+\alpha_j\bm{p}_j$ \algorithmiccomment{Update solution} \\
    $\bm{r}_{j+1}\gets\bm{r}_j-\alpha_j\bm{v}_j$ \algorithmiccomment{Update residual} \\
    \If{$\norm{\bm{r}_{j+1}}_2<\epsilon$}{
      \textbf{break}
    }
    $\beta_{j}\gets\frac{\bm{r}_{j+1}^*\bm{r}_{j+1}}{\bm{r}_{j}^*\bm{r}_{j}}$ \algorithmiccomment{Calculate the improvement factor} \\ 
    $\bm{p}_{j+1}\gets\bm{r}_{j+1}+\beta_j\bm{p}_j$ \algorithmiccomment{Update search direction}
  }
  \Return{$\bm{x}_{j+1}$.}
\end{algorithm}